\documentclass[journal]{IEEEtran}
\ifCLASSINFOpdf
  \usepackage[pdftex]{graphicx}
\else
\fi
\usepackage{amsmath}
\usepackage{algorithm}
\usepackage{algorithmicx}
\usepackage{algpseudocode}
\usepackage{caption}
\usepackage{subcaption}
\usepackage{amsthm}
\usepackage{multirow}
\usepackage[flushleft]{threeparttable}
\usepackage{amsfonts}

\usepackage{array}
\newtheorem{definition}{Definition}
\newtheorem{theorem}{Theorem}

\begin{document}
%
\title{FERN: Fair Team Formation for Mutually Beneficial Collaborative Learning}
%
%
%

\author{Maria~Kalantzi,
        Agoritsa~Polyzou,
        and~George~Karypis
\thanks{M. Kalantzi and G. Karypis are with the Department of Computer Science and Engineering, University of Minnesota, Twin Cities,
MN, USA. E-mail: \{kalan028, karypis\}@umn.edu.}
\thanks{A. Polyzou is with the Massive Data Institute and Computer Science Department, Georgetown University, Washington, D.C., USA. E-mail: ap1744@georgetown.edu.}
}

\maketitle

\begin{abstract}
    Automated \textit{Team Formation} is becoming increasingly important for a plethora of applications in open source community projects, remote working platforms, as well as online educational systems. The latter case, in particular, poses significant challenges that are specific to the educational domain. Indeed, teaming students aims to accomplish far more than the successful completion of a specific task. It needs to ensure that all members in the team benefit from the collaborative work, while also ensuring that the participants are not discriminated with respect to their protected attributes, such as race and gender. Towards achieving these goals, this work introduces \texttt{FERN}, a fair team formation approach that
    promotes mutually beneficial peer learning, dictated by protected group fairness as \textit{equality of opportunity} in collaborative learning.
    We formulate the problem as a multi-objective discrete optimization problem. We show this problem to be NP-hard and propose a heuristic hill-climbing algorithm. 
	Extensive experiments on both synthetic and real-world datasets against well-known team formation techniques show the effectiveness of the proposed method.
\end{abstract}


\IEEEpeerreviewmaketitle

\section{Introduction}
\IEEEPARstart{E}{xtensive} 
research has unveiled the benefits of collaborative learning through participation in group educational activities (e.g., study groups, team projects, etc.)~\cite{ashman2003cooperative, slavin1988cooperative, hertz2013learning, freeman2014active, bargh1980cognitive, hertz1995interaction}, especially when the members of the teams can learn something from their peers (peer learning)~\cite{hertz1995interaction, vygotsky1980mind, bargh1980cognitive, freeman2014active}.
This stands for traditional brick-and-mortar educational institutions, as well as online learning environments and massive open online courses. To facilitate the creation of such collaborative learning groups, \emph{team formation} algorithms are used to group students based on different factors and constraints. 
Since students have the opportunity to benefit through peer learning, it is important for the team formation algorithms to ensure that the students are treated fairly with respect to benefit from team work, irrespective of any protected attributes, such as gender and race.

There has been considerable research in developing methods to form teams in the educational domain, some of which try to account for different notions of fairness.
Some works implicitly address fairness when they aim to produce balanced teams with respect to the skills of their participants (inter-team homogeneity)~\cite{andrejczuk2019synergistic, bandyopadhyay2018divgroup, liu2016collaborative}, while others also account for diversity in personality and gender~\cite{andrejczuk2019synergistic}.
Outside the educational domain, recent work on fair clustering methods considers equal representation of the protected groups across the clusters~\cite{chierichetti2017fair, bera2019fair, ziko2019clustering}. 

In this work, we are motivated by the \emph{equality of educational opportunity}~\cite{sep-equal-ed-opportunity}, where every student should have equal educational opportunities irrespective of race, gender, socioeconomic class, sexuality or religion. We transfer this ideal of \emph{equality of opportunity} in the context of collaborative learning and we strive to ensure that all students benefit equally from peer learning, regardless any protected attributes.
Our goal is to create teams comprising students with the necessary set of skills for performing the target task, while maximizing the peer learning opportunities at both \textit{individual} and \textit{protected group} level. We consider cases where individuals can belong to different groups, based on a protected attribute such as race, demographics and socioeconomic status, and we focus on alleviating potential discrimination against the team members.
Our approach is formalized as a multi-objective discrete optimization problem that captures the aforementioned objectives. We formally prove its NP-hardness, and we propose \texttt{FERN}; a heuristic hill-climbing greedy algorithm to tackle it. We experimentally evaluate the performance of \texttt{FERN} on both synthetic datasets with different characteristics, and real-world datasets, against well-known team formation techniques. 
The results show the effectiveness of the proposed method on creating fair and beneficial teams.
To the best of our knowledge, this is the first work in the team formation domain that addresses group fairness in the form of equality of opportunity in collaborative learning.

\section{Problem Statement}
Given a set $\mathbb{S}$ of $N$ students and a target task with specific skill requirements, we are interested in assigning the students to teams such that, in addition to completing the target task, every member in a team benefits. 
Following existing conventions~\cite{hertz1995interaction, vygotsky1980mind, agrawal2014grouping, agrawal2017grouping}, we consider a student to benefit from participating in a team if there is at least one skill they can learn from their higher ability peers.
In this case, we can create teams such that even high performing members could benefit, as long as there is a peer who has higher ability in at least one skill.   
In our setting, the number of teams and their size are not fixed and depend on the input data and the target task.
We assume that a task requires multiple skills in order to be completed successfully. For example, different lower level courses could be required to complete a project in a higher level course.

\textbf{Skill requirements}.
A given target task with $k$ required skills is represented by a $k$-dimensional skill requirements vector, $\mathbf{r}=(r_1,\dots,r_k)$, whose $p$-th component, $r_p$, with $p\in \{1,\dots, k\}$, is the threshold that the students from each team need to reach collectively, in order to complete the task with respect to skill $p$.
Note here, that the requirements can be different. Large value for $r_p$ indicates that skill $p$ is of higher importance or that it requires higher ability from the students in order to complete the target task.
Student $i$ is represented by a $k$-dimensional vector, $\mathbf{s}^i = (s_1^i,\dots,s_k^i) $, whose $p$-th element, $s_p^i$, with $p\in \{1,\dots, k\}$, represents the student's ability for skill $p$. 
We consider a team $l$ to complete the task successfully when the aggregated ability of its members is greater than or equal to the corresponding skill threshold value, for all the $k$ skills, i.e.:
\begin{displaymath}
	\sum _{i \in \mathbb{T}_l} s^i_p \geq r_p, \forall p=1,\dots,k,
\end{displaymath}
where $\mathbb{T}_l$ is the set of students that belong to team $l$. 

\textbf{Individual benefit}.
A student $i$ benefits from a student $j$ when $j$ has higher ability in at least one skill. We formally define this as follows: 
\begin{definition}
    A student $i$ with skill vector $\mathbf{s}^i$ \textbf{benefits} from a student $j$ with skill vector $\mathbf{s}^j$, when there is at least one skill, $p \in \{1,\dots,k\}$, such that $s^j_p - s^i_p > \varepsilon$.
\end{definition}
We use a threshold $\varepsilon$ to distinguish among \textit{significant} and \textit{non-significant} differences in the skill values of the students.
Next, we define the benefit matrix, $\mathbf{B} \in \mathbb{R} ^{N \times N}$, with $B_{ij}=1$ if student $i$ benefits from student $j$, and $B_{ij}=0$, otherwise. 
Fig. \ref{fig:benefit_example} presents a toy example which highlights the definition of benefit. There are three students A, B, C, and their corresponding performance in three skills. If we group B with A then, A will learn skills 1 and 2 from B, and B will learn skill 3 from A. On the contrary, if we group B with C, student B will not benefit from C as B is performing better in every skill.
\begin{figure}[!t]
\centering
    \includegraphics[width=0.6\linewidth]{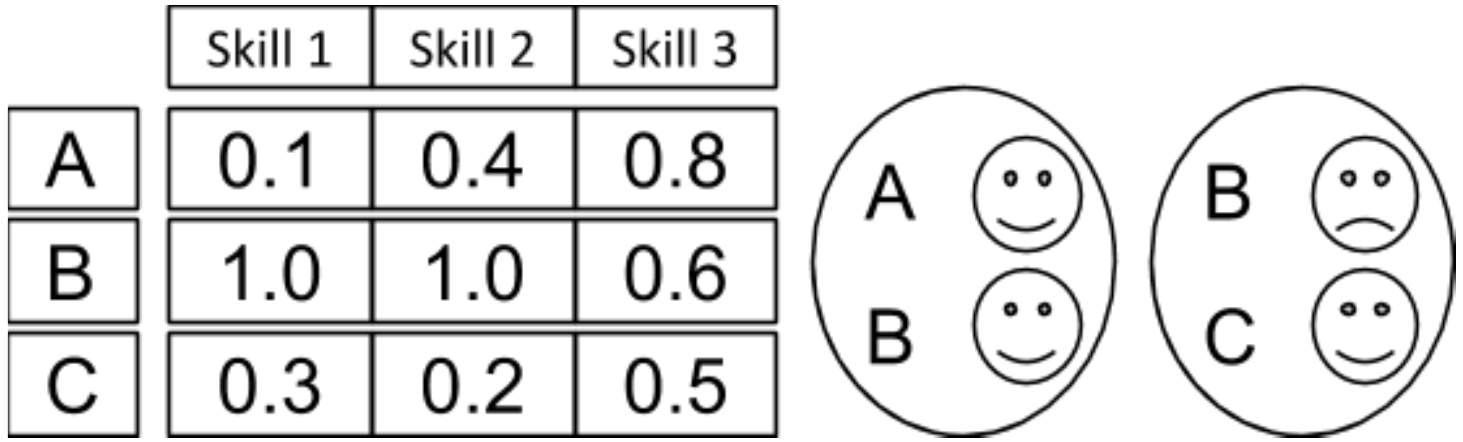}
    \caption{Example of the benefit of students A, B, C if we team them in two possible ways.}
    \label{fig:benefit_example}
\end{figure}

We then define the \textit{individual benefit} of student $i$ from their team, as the fraction of their teammates that $i$ benefits from, i.e.:
\begin{equation}
    \textit{IndBen}(i)= \frac{1}{|\mathbb{T}^i|}\sum_{j \in \mathbb{T}^i} B_{ij},
    \label{eq:indben}
\end{equation}
where $\mathbb{T}^i$ is the set of students that belong to the same team as $i$, excluding $i$. According to this definition, in order to increase the individual benefit of student $i$, we have to increase the number of teammates that $i$ benefits from. In other words, we aim to create teams in which the members benefit from as many teammates as possible.

\textbf{Group benefit}.
We assume that the students belong to $m$ different groups based on a protected attribute, such as race, demographics and socioeconomic status. The benefit of a protected group $q$, with $q\in \{1, \dots,m\}$, is the average individual benefit of its members, i.e.:
\begin{equation}
    \textit{GBen}(q)= \frac{1}{|\mathbb{G}_q|}\sum_{i \in \mathbb{G}_q} \textit{IndBen}(i),
    \label{eq:gben}
\end{equation}
where $\mathbb{G}_q$ is the set of students that belong to protected group $q$.
We define \textbf{group fairness} in the context of equality of opportunity in collaborative learning and we consider a team solution to be fair across the protected groups, if the different protected groups have approximately the same relative group benefit. In particular, for two protected groups $q_1$ and $q_2$, a team solution is fair when:
\begin{equation*}
    \textit{GBen}({q_1}) \approx \textit{GBen}({q_2})
\end{equation*}
If the students that tend to benefit the least disproportionately belong to one protected group, then this corresponds to an unfair team solution.

In summary, we aim to form teams that fulfill the skill requirements of the target task, maximize the individual benefit and impose group fairness. This problem is a combinatorial optimization problem, with a large discrete configuration search space.
Next, we show that this is an NP-hard problem.

\begin{theorem}
	The problem is NP-hard.
\end{theorem}
\begin{proof}
	In order to prove that our problem is NP-hard, we will reduce the Dual Bin Packing Problem (DBPP)~\cite{assmann1984dual} to a much simpler version of our problem.
	We consider the case where we have only one skill for each student and the goal is to assign the students to teams, such that, the number of teams which fulfill or exceed the skill requirements is maximized.
	We then map the bins to teams and the items to students. 
    The skill $s^i$ of student $i$ corresponds to the size $\alpha_i$ of item $i$ and the bin threshold $t$ corresponds to the skill requirement threshold $r$. We want to find the maximum number of teams, $L$, and a partition of the students to teams, $\mathbb{T}_1\cup \dots \cup \mathbb{T}_L$, such that $\sum_{i \in \mathbb{T}_l}s^i \geq r, \forall l=1,\dots,L$.
	Thus, the DBPP is reduced to our problem and if we can find a solution in polynomial time to our problem, then we can find a solution to DBPP as well. But we know this is a contradiction since DBPP is NP-hard~\cite{assmann1984dual}.
\end{proof}

\section{FERN: A Fair Team Formation Approach}

We formulate the problem of fair team formation as a multi-objective discrete optimization problem and we present a heuristic algorithm to solve it. We will use the name \texttt{FERN} to refer to the approach and associated algorithm. \texttt{FERN} starts by computing an initial solution and then refines this solution to gradually improve the objective function.

\subsection{Objective Function}
Our goal is threefold: create teams such that 1) the collective ability of each team is adequate to successfully complete the target task, 2) the number of students that benefit from the team work is maximized, and 3) impose fairness among the protected groups. 
Next, we explain how we proceeded in order to achieve each goal.

\textbf{Average skill deficiency}. 
A team completes the target task successfully, if its collective ability reaches or exceeds the skill requirements threshold. In order to tackle this goal, we measure the \textit{average skill deficiency} of a solution, which is given by:
\begin{equation}
    \mathcal{X} = \frac{1}{L k} \sum_{l=1}^L \sum_{p=1}^k \bigg(r_p - \min\big(r_p,\sum_{i\in \mathbb{T}_l} s_p^i\big)\bigg)^2,
    \label{eq:capacity_term}
\end{equation}
where $L$ is the number of teams of the solution.

\textbf{Average individual benefit}. In order to tackle our second goal and maximize the number of students that benefit from collaborating with others in the same team, we measure the \textit{average individual benefit} over the whole set of students $\mathbb{S}$:
\begin{equation}
   \mathcal{Y} =\frac{1}{|\mathbb{S}|} \sum _{i\in \mathbb{S}} \textit{IndBen}(i),
    \label{eq:avgindben}
\end{equation}

\textbf{Variance in group benefit}. In order to achieve group fairness with respect to benefit among the $m$ protected groups, we measure the \textit{variance in group benefit}, which is given by:
\begin{equation}
    \begin{aligned}
    \mathcal{Z} = \frac{1}{m} \sum_{q=1}^{m} \Big(\textit{GBen}(q) - \overline{\textit{GBen}} \Big)^2,\\
    \text{where } \overline{\textit{GBen}} = \frac {1} {m} \sum_{q=1}^{m} \textit{GBen}(q).
    \end{aligned}
    \label{eq:fairness}
\end{equation}{}

\textbf{Combined objective function}. At last, we combine Eq.~\ref{eq:capacity_term}, \ref{eq:avgindben}, \ref{eq:fairness} in a single multi-objective function and we solve the following minimization problem:
\begin{equation}
    \mathcal{F} = \mathcal{X} - 
    \gamma \mathcal{Y} +
    \delta \mathcal{Z},
    \label{eq:obj}
\end{equation}
where we want to minimize $\mathcal{X}$, $\mathcal{Z}$ and maximize $\mathcal{Y}$. We use parameters $\gamma$ and $\delta$ to control the relative importance of individual benefit and group fairness, respectively. 

The group benefit is the most novel component of our multi-objective function.
The motivation behind our definition of group benefit (Eq.~\ref{eq:gben}) and the way we measure group fairness (Eq.~\ref{eq:fairness}) is the following. 
The second term of the objective function (average individual benefit, $\mathcal{Y}$) maximizes the individual benefit of each student and implicitly ensures that the number of students that benefit is maximized. This means that the group benefit of each protected group will also increase.
Then, the third term of the objective function (variance in group benefit, $\mathcal{Z}$) additionally takes care of the balance in group benefit across the protected groups. This corresponds to improving the least benefited group, since the second term will penalize any decrease in the benefit. 
Ultimately, our combined formulation is similar to the one of maximizing the group benefit of the least benefited protected group. To back up the above arguments, we present experimental results in Subsection~\ref{sec:motivation}.

\subsection{Initial Assignment Algorithms}
\label{sec:initialAssignment}
We developed four different algorithms to compute the initial assignment of students to teams, which are used as the starting point for the refinement process.

\subsubsection{Global Most-Benefit First (\texttt{GMBF})} 
\label{gmbf}
This method creates one team at a time; it starts with the first team and keeps on adding students until the skill requirement threshold is reached for each skill, or there are no students left to select. Every time, it chooses student $i$ who has not yet been selected and has the highest individual benefit if all the students were teamed up in a single team, i.e., $i=\arg \max_i \sum_{j\in \mathbb{S}} B_{ij}$ - \textit{global} choice. 
This algorithm optimizes the average skill deficiency (the first term of the multi-objective function) and ensures that every team, except for probably the last one, will reach the threshold in order to complete the target task. 

\subsubsection{Local Most-Benefit First (\texttt{LMBF})} This method is similar to \texttt{GMBF}, with the difference now, that it chooses the student with the highest individual benefit based on the current team it creates - \textit{local} choice. 
The first student of each team is selected to be the one with the global lowest individual benefit out of those not yet selected. This algorithm accounts for the first two terms of the multi-objective function at every local choice it makes.

\subsubsection{Local Most-Benefit Fair First (\texttt{LMBFF})} This method extends the previous two and takes care of all terms of the multi-objective function. The difference compared to the two aforementioned algorithms is that, now the algorithm chooses to add to the current team, that student who optimizes the weighted summation of the average individual benefit and the variance in group benefit. Again, the first student of each team is selected to be the one with the global lowest individual benefit out of those not yet selected. 

\subsubsection{Random} \label{sect:Random} This corresponds to a random assignment and serves as a baseline to compare against the above more sophisticated proposed initial algorithms.  
It randomly assigns equal number of students to each team. If this is not possible, the remaining students are equally split in the first $N$ \verb|modulo| $L$ teams, where the number of teams $L$, is determined by \texttt{GMBF}.

\subsection{Refinement Algorithms}

We developed two hill climbing greedy approaches to perform the refinement process. They are best-first iterative algorithms that start from an initial sub-optimal solution, and attempt to find a better solution by making an incremental change at a time to the current solution. This is done repeatedly, until no further improvement can be found. 

The change consists of moving student $x$ from their current team $T_{s}$ to team $T_{d}$; at each step, we choose the triplet $(x, T_{s}, T_{d})$ that gives the highest gain (optimal available move).
The gain $\mathcal{G}_{x,T_{s},T_{d}}$ of a move $(x,T_{s},T_{d})$ is the reduction on the objective if we move student $x$ from their current team $T_{s}$ to team $T_{d}$, i.e.,  
	$\mathcal{G}_{x,T_{s},T_{d}} = \mathcal{F}_{x\in T_{s}}-\mathcal{F}_{x\in T_{d}}$.
The gain is positive if the move improves the value of the objective function, negative if the move deteriorates the solution, and equal to zero if the gain does not change.

The first algorithm \texttt{SAHC}, is a greedy steepest-ascent hill climbing, that searches the whole space of possible moves and performs the best one which leads to a \textit{strictly positive} gain; if no such move exists, it terminates. The pseudo code is presented in Algorithm \ref{alg:sahc}. 

The second algorithm is based on the Fiduccia-Mattheyses (FM) algorithm~\cite{fiduccia1982linear}; an iterative heuristic for hyper-graph bi-partitioning. This FM-based hill climbing, \texttt{FMHC}, is different from the steepest-ascent as it allows for uphill moves, i.e., moves that do not improve the value of the objective function, and even worsen it. 
It may perform \textit{bad} moves with the goal of escaping from \textit{local minima}. 
In particular, at each pass, \texttt{FMHC} examines all possible moves and greedily chooses the one with the optimal gain (it could be a negative one), until there are no more students to move. 
After it makes a move, it \textit{locks} the student and does not consider them again in the same pass. It considers again the locked students in the next pass, and only then, it updates their gains. Next, \texttt{FMHC} finds the number of subsequent moves $z$, from the beginning of the current pass, that maximizes the sum of the gains of the $z$ moves, $gain_{max}$. This translates to identifying a set of $z$ concurrent moves (including potential non-beneficial single moves) that altogether will lead to a better solution. If $gain_{max}$ is larger than $\epsilon$, then it accepts the $z$ moves, updates the assignment of students to teams and improves the solution. If $gain_{max}$ is smaller than or equal to $\epsilon$, then the refinement terminates. 
The pseudo code is presented in Algorithm \ref{alg:fm}.

For both algorithms, we implemented a max-heap priority queue to store the gains efficiently. During the refinement, the algorithms are allowed to empty a team and consequently, reduce the number of teams, if this leads to a better solution. At the end of the refinement process, we perform a post-processing step to eliminate teams with only one student; we assign the student to the best team, based on the value of the objective function.   

\begin{algorithm}[bt] 
	\caption{Steepest Ascent Hill Climbing (\texttt{SAHC})}
	\label{alg:sahc}
	\begin{algorithmic}[1]
		\Require initial assignment $T$, initial number of teams $L$.
		\State $pq \gets$ compute gains \Comment{priority queue with gain values for all possible moves}
        \While {1}
		\State $\textit{gain, student, }team_{s}\text{, }team_{d} \gets pq.pop()$ \Comment{move with the maximal gain}
		\If {$gain > 0$} 
		\State $\text{Perform the move and update } T$
		\State{Update gains $\forall student_i \in S$}
		\Else 
		\State $\textbf{break}$
		\EndIf
		\EndWhile
		\State $\text{Remove empty teams from } L$. Check for teams with only one student and find the best other team to assign this student to.
		\State \Return $T, L$
	\end{algorithmic}
\end{algorithm}

\begin{algorithm}[bt]
	\caption{FM-based Hill Climbing (\texttt{FMHC})}
	\label{alg:fm} 
	\begin{algorithmic}[1]
		\Require initial assignment $T$, initial number of teams $L$.
		\State $pq \gets \text{[ ]}$ \Comment{priority queue that stores the gain values for all possible moves}
		\While {$1$}
		\State{$pq \gets$ compute gains}
		\State{$locked \gets set()$, $T_{s} \gets T$, $moves \gets \text{[ ]}$}
		\While {$pq$} \Comment{while there are students to move}
		\State $\textit{gain, student, }team_{s}\text{, }team_{d} \gets pq.pop()$ \Comment{move with the maximal gain}
		\State $locked.add(student)$
		\State{Remove from $pq$ all entries that correspond to $student$}
		\State $\text{Update }T_{s} \text{ as if we make the move}$
		\State{Update gains $\forall student_i \in (S-locked)$}
		\State{$moves.append([[student, source, dest], gain])$}
		\EndWhile
		\State $z \gets 0$, $gain_{max} \gets -\inf$
		\State{Find \textit{z} that maximizes $gain_{max}$, the sum of $moves[1][1],\dots, moves[z][1]$}
		\If {$gain_{max} > \epsilon$} 
		\State $\text{Update } T \text{ with the } z \text{ moves}$
		\Else 
		\State $\textbf{break}$
		\EndIf
		\EndWhile 
		\State $\text{Remove empty teams from } L$. Check for teams with only one student and find the best other team to assign this student to.
		\State \Return $T, L$
	\end{algorithmic}
\end{algorithm}

\textit{Time Complexity Analysis.}
The most computationally expensive part in both the refinement algorithms is the update of the gains in the priority queue after we make a move. In the priority queue, we store in total $N\times (L-1)$ moves for the $N$ students and $L$ teams. 
The update of all the gains costs $O(N L \log (N L))$. Thus, the time complexity of \texttt{SAHC} for each iteration of the \texttt{while} loop is $O(N L \log (N L))$. For \texttt{FMHC}, the overall time complexity of one pass is $O(N^2 L \log (N L))$, since the inner \texttt{while} loop (line 5) is executed $N$ times.

\section{Related Work}
\label{sect:RW}

Agrawal et al.,~\cite{agrawal2014grouping} were the first to introduce a computational approach for team formation in the educational setting. 
They form teams of \textit{leaders} and \textit{followers} (those with ability above and below the team average, respectively) to maximize the total gain.
They measure the gain as the number of students that can increase their ability by interacting with leaders. 
Another work by Agrawal et al.,~\cite{agrawal2017grouping} is built upon the intuition that every student can increase their benefit up to a reference score and extends the previous work~\cite{agrawal2014grouping} which assumes that only students below the team mean can benefit. In such formulation, every student could potentially increase their gain except for the best performing students of each team.
Both these works consider one-dimensional skill vectors, do not incorporate the notion of the skill requirements and are limited to maximize the benefit of a subset of students.

Bandyopadhyay et al.,~\cite{bandyopadhyay2018divgroup} propose a method based on the Modularity Optimization technique, used in community detection to group a collection of objects. Their goal is to maximize the intra-team diversity and minimize the inter-team diversity, by maximizing the pairwise distance between the individuals. 
Liu et al.,~\cite{liu2016collaborative} propose a method to extract features for students' skill proficiency by performing a cognitive analysis, and use these features to create the students' vectors. They also propose the \textit{uniform $k$-means} method for creating a fixed number of teams $K$, that maximizes collaborative learning. This approach uses a variant of $k$-means clustering method to create $k$ equal-sized clusters based on the students' vectors and then, randomly distribute the students of each cluster to $K$ teams. The goal is to create heterogeneous teams by maximizing an objective based on dissimilarities.
Methods like these, which promote the intra-team distances, could end up hurting both the individual-level and potentially the protected group-level benefit (unfair solution): the high ability students will be teamed up with lower ones, and the former will not benefit.   
Our objective function is entirely different from the aforementioned techniques, and the proposed approaches cannot be directly applied to solve our problem.

Bahargam et al.,~\cite{bahargam2019guided} defined the Guided Team-partitioning problem; given a set of points, they want to partition them into teams such that the centroid of each team is as close as possible to a specific target vector. They consider multiple target tasks with different requirements. 
This work is different from ours in that there are multiple target tasks and the team size and number of teams are fixed. The major difference though is that it allows for the removal of students out of the assignment solution. This is not allowed in our problem and cannot be applied in real-world educational settings.

Other works take into account the social network of the members in addition to the skill values, such that the communication cost is minimized and they are interested in forming a single team for the target task~\cite{lappas2009finding, scarlat2011genetic, niveditha2017genetic}.
Andrejczuk et al.,~\cite{andrejczuk2019synergistic} consider the problem of forming \textit{synergistic} teams, where each team is balanced in terms of ability, personality, and gender. The final equal-sized teams are diverse (equal representation) and of equal performance. One of the proposed methods is a local search algorithm that is similar to our hill-climbing method, but with significant differences: 1) their approach is not greedy, 2) it performs different kind of moves than ours, and 3) it allows for only a limited number of consecutive bad moves. Finally, they optimize a different objective function as they do not account for skill requirements or fairness in the form of equality of opportunity (we point out that the latter stands for all the related work), which makes this method unsuitable to solve our problem.
Machado et al.,~\cite{machado2019fair} also account for fair teams, but in the notion again of forming teams of equal performance: the best-skilled members are equally distributed across the teams. Another difference with our work is that, even though they too build their problem upon members' skills and project requirements, they do not use numerical values; instead, they use text tags and assume that a single member can fulfill one skill. Last, they target multiple tasks and assign the best team to each task. 

Genetic algorithms have also been proposed for solving the team formation problem~\cite{moreno2012genetic, scarlat2011genetic, niveditha2017genetic, mishra2019worker}. They are used for optimization problems and they incorporate the idea of evolution. They model the solutions as chromosomes and apply recombination operations (crossover and mutation) in order to preserve only the necessary information that will lead to a better solution~\cite{moreno2012genetic}.

\textit{Crowd-sourcing}.
Another related field is that of the standard crowd-sourcing. Within that context, given a set of workers with specific skills and cost, the goal is to form a single team for a target task.
Mishra et al.,~\cite{mishra2019worker} study a variation of the problem which does not include any cost on the workers' side. They use the notion of \textit{dominance} to form a ranking of workers out of a pool of available workers, in order to choose the best ones for the target task, that are non-dominated with each other, i.e., they are equally preferable. 
Their definition of dominance is related to our definition of benefit; in our context, we want to minimize the dominance in order to increase the benefit. However, 
we are interested in allocating all the individuals to teams, so that the task's skill requirements is reached.
Yadav et al.,~\cite{yadav2017concurrent} study the problem of splitting the workers into teams that satisfy some requirements, but each team is assigned a different task. The minimization objective is related to each worker's cost.

\section{Experimental Evaluation}

\subsection{Datasets}

We evaluated the performance of \texttt{FERN} on two sets of datasets. The first is a set of synthetically generated datasets of different size, complexity, and solution difficulty. The second is a set of datasets that we derived from actual students' course registration and performance data. Details on how we generated the synthetic and real-world datasets are provided next.

\subsubsection{Synthetic Datasets}
We developed a synthetic dataset generation method with five different parameters: (i) the number of students $N$, (ii) the dimension of the skill vectors $k$, (iii) the number of protected groups $m$, (iv) a vector of length $m$ specifying how the $N$ students are distributed across the $m$ groups and, (v) $m$ sets of parameters that control how the skill vectors for students of each protected group are generated.  
Next, we describe the last type of parameters, (v).

We assume that the completion of a target task requires knowledge of specific courses and the ability of a student to complete the task is represented by the corresponding grades (skill values). 
We generate the data in such a way that the students fall into four \textit{grade buckets} reflecting their overall average performance: $A$, $B$, $C$, and $D$.
Based on our generation model, for each protected group, we follow two steps: first, we distribute the students to the grade buckets, and second, we generate the skill vectors of each grade bucket. 

In the first step, we want to create the grade buckets' distributions. Our goal is to generate datasets of progressively increased difficulty. The shape of these distributions affects the difficulty of balancing the generated dataset with respect to benefit. 
When the distribution of the students to the grade buckets is different across the protected groups, then students from different groups have different performance. For example, in the case of two protected groups, most of the students of the first group could belong to bucket $A$ and most of the students of the second group to bucket $C$. In such cases, it is harder for the algorithm to impose group fairness.
To this end, we sample from a beta distribution and we create multiple datasets of increasing difficulty by varying the values of $\alpha$ and $\beta$ parameters.
We chose beta distribution as it is able to model grade distributions~\cite{arthurs2019grades}, can produce several different shapes, and is suitable to model random percentages. By definition, the beta distribution is defined on the interval $[0,1]$; in order to correspond the sampled values to the four grade buckets, we discretize $[0,1]$ to four bins. 
Table~\ref{tab:datasets} summarizes the statistics of the generated data for two protected groups.
Fig.~\ref{fig:datasets_distr} presents the shapes of the distributions for the synthetic data when the number of protected groups is equal to $2, 3, 4$. The different protected groups are of equal size.
Dataset D1 is considered an easy dataset, as the different protected groups have the same bucket distributions, i.e., the students of different protected groups have similar performance.
Dataset D2 is of middle difficulty, as the bucket distributions of the two groups are slightly different from each other. Finally, dataset D3 is the hardest one and represents an extreme case, where the two distributions are opposite with each other.
In this case, it will be very difficult to impose group fairness, as there will be many high-skilled (bucket $A$) members of one protected group who are difficult to benefit from the lower-skilled (bucket $D$) members of the other group. 

\begin{table}[!t]
    \renewcommand{\arraystretch}{1.3}
    \begin{threeparttable}
	\caption{Statistics of the synthetic datasets for two protected groups.}
	\label{tab:datasets}
 	\setlength{\tabcolsep}{0.35em}
 	\centering
	\begin{tabular}{lrrrrrr||rrrrrr}
		\hline
		& \multicolumn{6}{c}{Group 1}&\multicolumn{6}{c}{Group 2} \\
		\hline
		Dataset & $A$ & $B$& $C$& $D$ & $\alpha$ &$\beta$ & $A$ & $B$ & $C$ & $D$ & $\alpha$ &$\beta$ \\
		\hline
		D1 & 16.2 & 57.0 & 25.7 & 1.1 &6.0 &4.0 & 16.2 & 57.0 & 25.7 & 1.1 &6.0 &4.0\\		
		D2 & 42.3 & 51.0 & 6.6 & 0.1& 8.0 &3.2 & 7.4 &	59.4 & 32.2 & 1.0& 7.0 &5.5 \\
		D3 &87.7 &11.8 &0.5 &0.0 & 7.5 &1.0 &0.0 &0.6 &11.1 &88.3&1.0 &7.5 \\
		\hline
	\end{tabular}
	\begin{tablenotes}
	\item $A, B, C, D$ are the grade buckets. The values of each bucket correspond to average percentages (\%) of the students that belong to each particular bucket. $\alpha$, $\beta$ are the parameters of the beta distribution used for each dataset. The two protected groups are of equal size. 
	\end{tablenotes}
	\end{threeparttable}
\end{table}
\begin{figure}
    \centering
    \includegraphics[width=\linewidth]{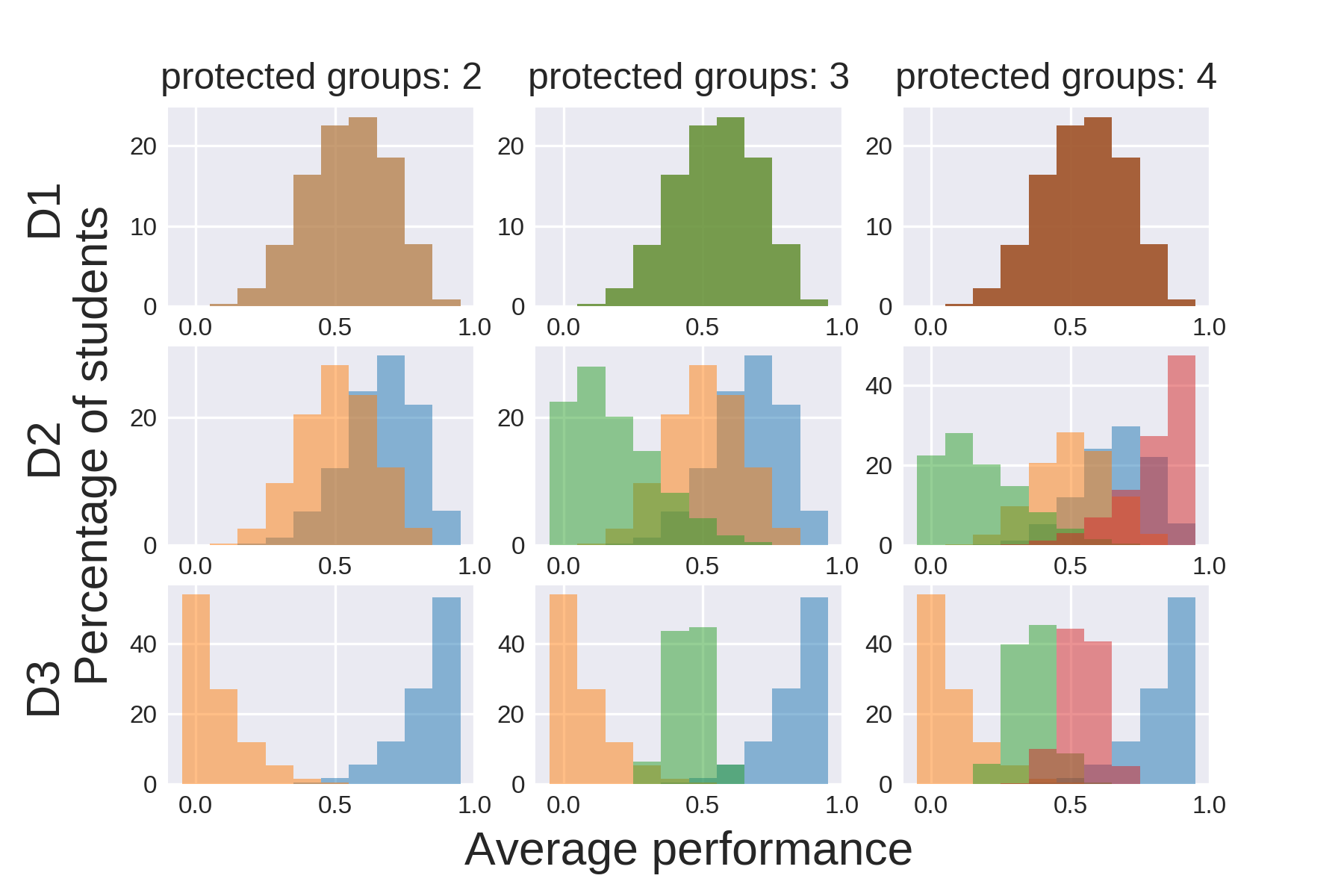} 
    \caption{Synthetic data - distribution of the students to the four grade buckets based on their average performance.}
    \label{fig:datasets_distr}
\end{figure}

In the second step of the data generation process, we generate the students' skill vectors of each grade bucket.
Since the skill values correspond to course grades, we map the letter grades $A-D$ to numeric ones. 
Then, we sample values from a normal distribution and we choose the mean to be the middle point of the corresponding numeric grade set of each grade bucket: 3.85 for $A$, 3.0 for $B$, 2.0 for $C$ and 1.15 for $D$. In all cases, we set the variance equal to $0.1$. Last, we normalize the generated skill vectors, so that the values belong to the interval $[0,1]$.

\subsubsection{Real-world Datasets}
We collected data from three departments of a large public university: Mathematics (MATH), Computer Science (CSCI), and Mechanical Engineering (MENG). The data include the grades of undergraduate students and span a period of 12 years.
We selected the students who had received an A-D grade in the $k$ most frequently taken courses from each major. 
We used the corresponding grades as the skill values for each student. Next, we transformed the letter grades to numerical values and normalized them in the range $[0,1]$. We used 2, 4, and 6 courses, i.e., $k=2,4,6$. In order to be able to explore what is the effect of the different values of $k$ on the solution, we had to make sure that the set of students will be the same as we change $k$. Thus, we selected the students obtained when $k=6$ (which is a subset of the students when $k=2,4$) and used those for the other values of $k$ as well.
For CSCI and MENG, where we had a large number of students, we selected the 400 most recent students, out of those that had taken the 6 most popular courses. For MATH and $k=6$, we had 355 students. 

The available data did not include any protected attributes.
In order to assign the students to protected groups, we used the entry registration status and the number of credits transferred to simulate the \textit{socioeconomic status} of the students. 
Our motivation came from the following findings.
High school students can take Advanced Placement (AP) courses and transfer the credits earned to their undergraduate program.
Minorities and low-income students are underrepresented in AP classes, and only a low percentage of them actually take and pass the AP exams every year~\cite{tugend2017benefits,whiting2009multi}.
Regarding students coming from other institutions, we did not have information about the institution they transferred from. However, reports statistically show that almost half of these students come from 2-year colleges~\cite{shapirotracking}, which are considered a major access point to 4-year institutions for minority and low-income students~\cite{crisp2014understanding}.
Finally, our analysis led to three protected groups: high school students with less than $15$ credits transferred (HS), high school students with more than $15$ credits transferred (HSAP), and those coming from other institutions/colleges (NAS). 
Table~\ref{tab:realdatasets} presents the statistics of the datasets for two courses ($k=2$). Fig~\ref{fig:real_distr} presents the distribution of the students to the four grade buckets based on their average performance in two courses. 

\begin{figure}[!t]
    \centering
	\begin{subfigure}{0.16\textwidth}
		\centering
		\includegraphics[width=\linewidth]{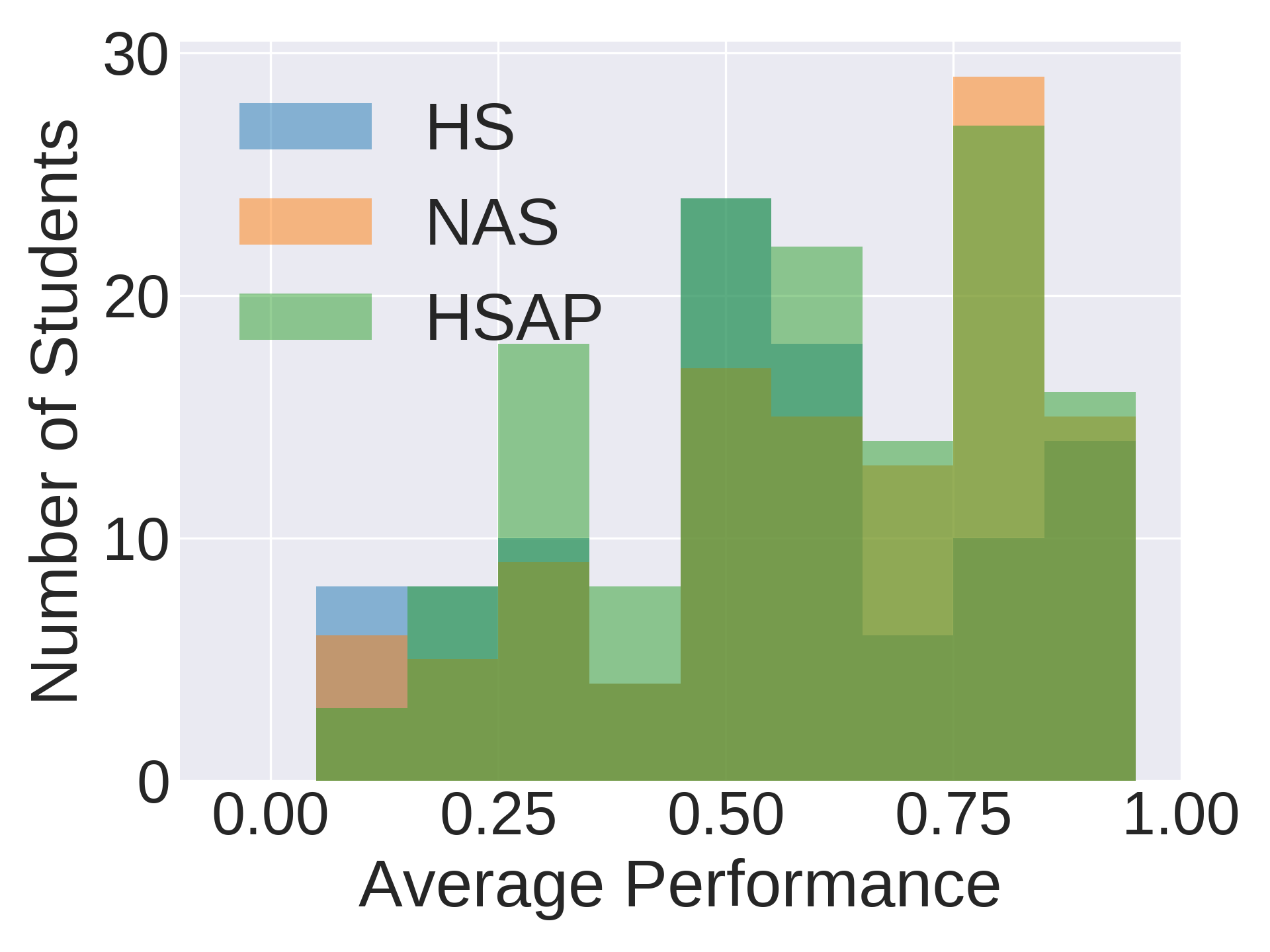}
		\caption{MATH}
		\label{fig:math_distr}
	\end{subfigure}%
	\begin{subfigure}{0.16\textwidth}
		\centering
		\includegraphics[width=\linewidth]{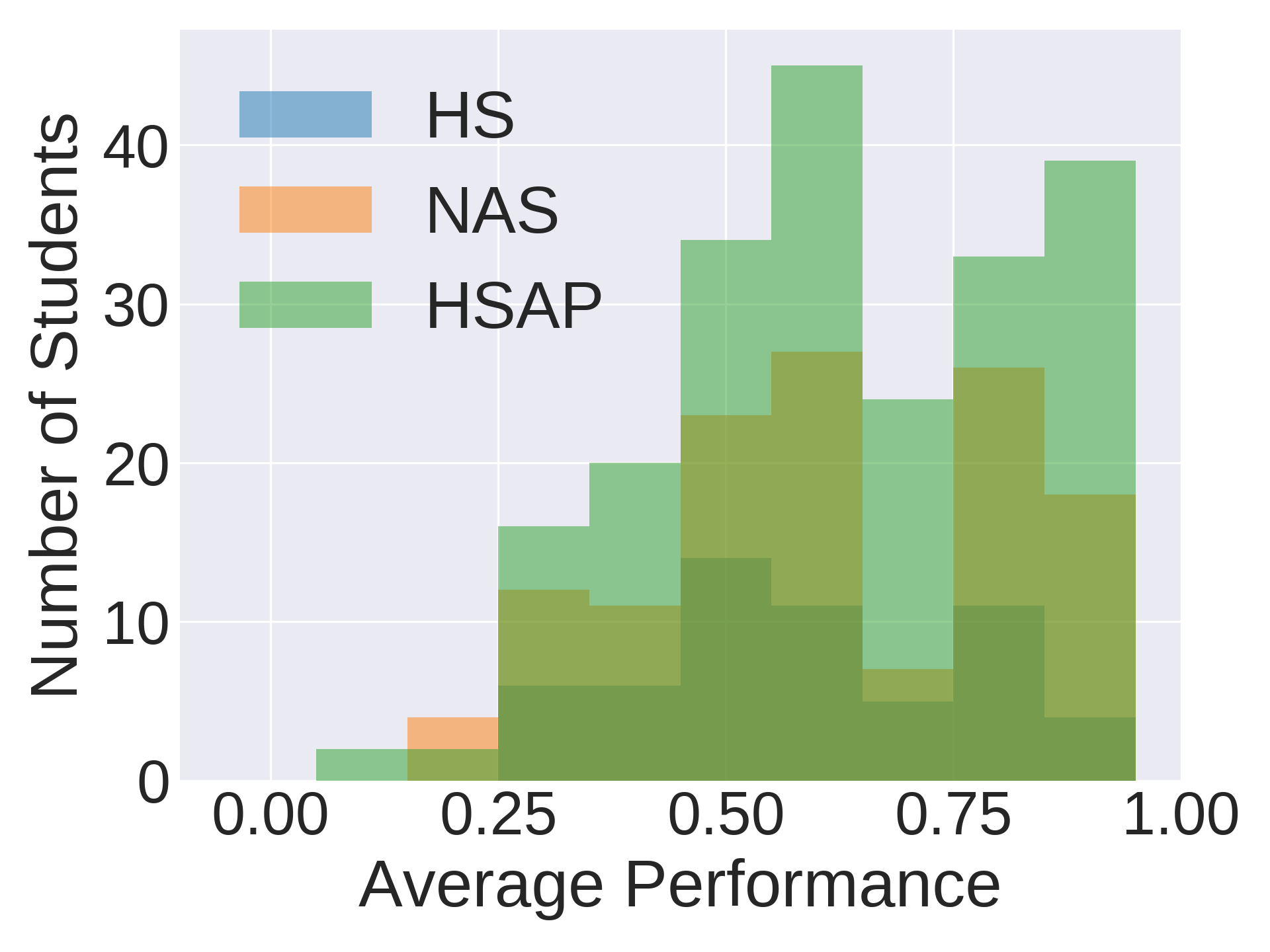}
		\caption{CSCI}
		\label{fig:csci_distr}
	\end{subfigure}%
	\begin{subfigure}{0.16\textwidth}
		\centering
		\includegraphics[width=\linewidth]{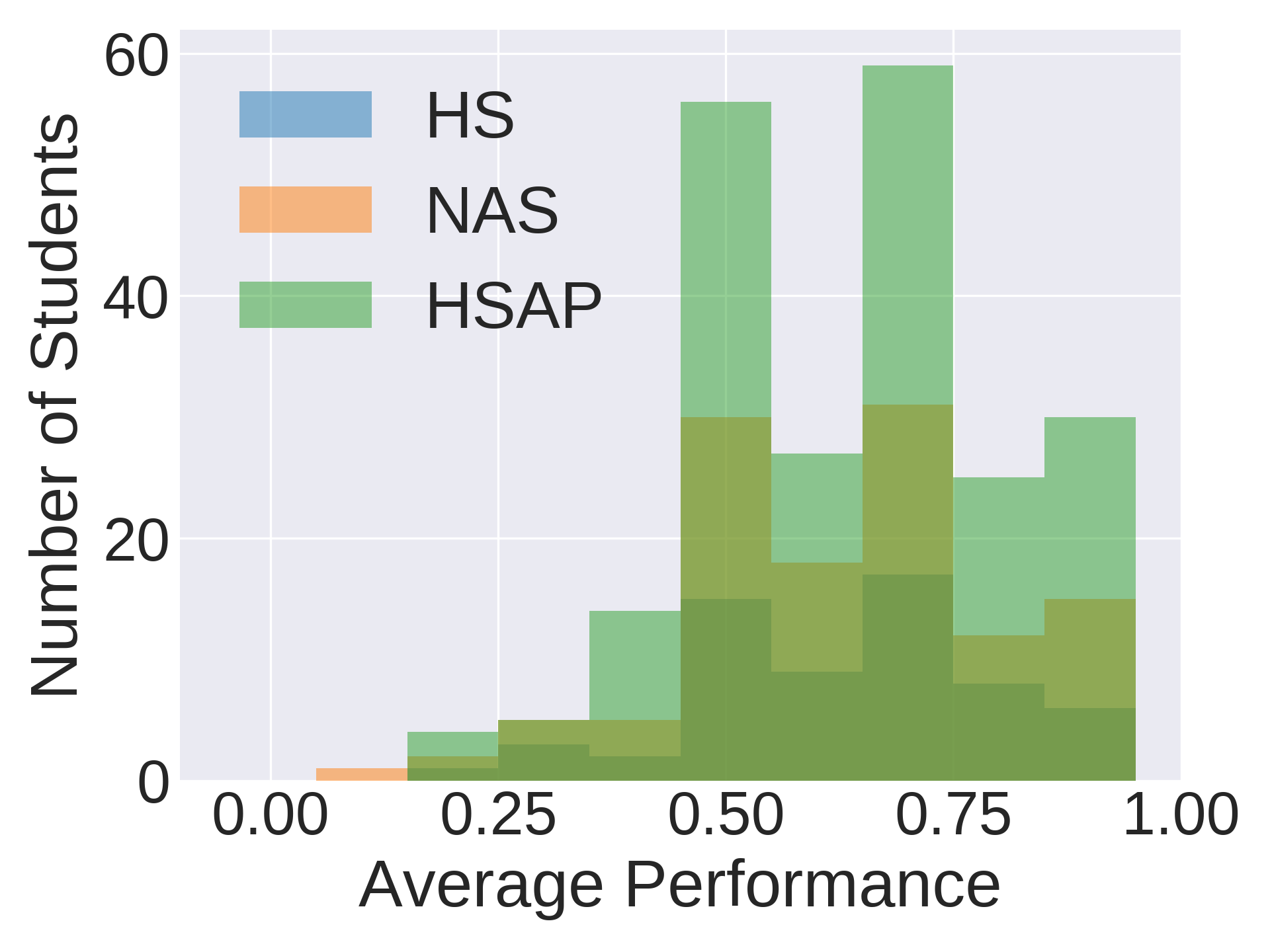}
		\caption{MENG}
		\label{fig:meng_distr}
	\end{subfigure}
	\caption{Real data - distribution of the students to the four grade buckets based on their average performance in two courses ($k=2$).}
	\label{fig:real_distr}
\end{figure}

\begin{table}[!t]
    \renewcommand{\arraystretch}{1.3}
    \begin{threeparttable}
	\caption{Statistics of the real-world datasets.}
	\label{tab:realdatasets}
	\setlength{\tabcolsep}{0.27em}
	\centering
	\begin{tabular}{lr|rrrrr|rrrrr|rrrrr}
		\hline
		\multicolumn{2}{c}{}& \multicolumn{5}{c}{HS} & \multicolumn{5}{c}{NAS} & \multicolumn{5}{c}{HSAP} \\
		\hline
		Dataset & $N$ &$N_1$ & $A$ & $B$& $C$& $D$ &$N_2$ & $A$ & $B$ & $C$ & $D$ & $N_3$ & $A$ & $B$ & $C$& $D$ \\
		\hline
		MATH & 355  & 102 & 30 & 42 & 22 & 8 &113 & 57 & 32 &18 & 6 & 140 & 57 &	46 & 34 & 3\\
		CSCI & 400  & 57 & 20	&25	&12	&0  &128	&51	&50	&27	&0  & 215 &96	&79	&38	&2 \\
		MENG & 400  & 61 & 19	&36	&5	&1  &119 &43	&63	&12	&1  & 220 &87	&110	&23	&0 \\
		\hline
	\end{tabular}
	\begin{tablenotes}
	\item $N$ is the total number of students and $N_1,N_2,N_3$ the number of students that belong to groups HS, NAS, and HSAP, respectively. $A, B, C, D$ are the grade buckets.
	\end{tablenotes}
    \end{threeparttable}
\end{table}

\subsection{Baseline Algorithms}
\label{sec:baselines}
We found that none of the available team formation algorithms that we listed in Section~\ref{sect:RW} can be applied directly to solve our problem. This is because none of them accounts for the equality of opportunity in collaborative learning, as we do. As a result, they optimize an entirely different objective function.
To this end, we implemented the following baseline approaches to compare against our proposed method \texttt{FERN}.

    \texttt{Random}. This method is described in Subsection~\ref{sect:Random} and corresponds to a random assignment of the students to teams. 
    
    \texttt{Uniform $k$-means (Umeans)}~\cite{liu2016collaborative}. This method (described in Section~\ref{sect:RW}) creates heterogeneous equally-sized teams by maximizing a dissimilarity-based objective. We create $\left \lceil{N/L}\right \rceil$ clusters, where $L$ is the number of teams generated by \texttt{GMBF}. This approach does not account for group fairness.
    
    \texttt{Genetic algorithm (GA)}. This approach is motivated by the genetic algorithms proposed for the team formation problem and described in Section~\ref{sect:RW}, and it is adapted to directly optimize our multi-objective function, Eq.~\ref{eq:obj}. We model the chromosome/solution as a vector of size $N$, i.e., the number of genes is $N$. We set the following parameters: mutation probability $= 0.1$, population size $ = 200$, number of generations $ = 300$, uniform crossover and swap mutation. For the number of teams $L$, we use the one that \texttt{FERN} identifies.
    
    \texttt{GMBF}. This algorithm is described in Subsection~\ref{gmbf} and is one of the algorithms that we experimented with to get an initial assignment. We do not use refinement process afterwards; this corresponds to a method like \texttt{FERN} without the refinement. This is an algorithm that optimizes the average skill deficiency (first term of the objective) and ensures that every team will reach the threshold, except potentially for the last one.

\subsection{Experimental Setting}
We try the following values for the importance parameters $\gamma \text{ and } \delta$: $\{0, 0.2, 0.5, 1, 2, 3\}$  
The gain threshold that we use as the convergence criterion for \texttt{FMHC} is set to $\epsilon=10^{-4}$. We also tried smaller values of $\epsilon$, i.e., $\epsilon=10^{-5}$ and $\epsilon=10^{-6}$, but the gains in performance were negligible, whereas the running time was significantly increased, as the algorithm performed significantly more work until convergence.
We run experiments for different number of students, i.e.,  $N=50,100,200,400$ for the synthetic datasets D1, D2, and D3. 

Unless otherwise stated, we use the following set of parameter values: number of dimensions for the skill vectors $k=2$, skill requirements threshold vector $\mathbf{r} = \mathbf{2}$, $\gamma = \delta =1$ which corresponds to equal importance among the three terms of the objective function, $\varepsilon=0$ for the benefit threshold, number of protected groups $m=2$ for D1, D2, D3 and $m=3$ for MATH, CSCI, MENG. For \texttt{Umeans} and \texttt{GA}, we run each experiment 10 times and we report the average performance. 
Finally, for each synthetic dataset, we run $50$ experiments with different sampled data for the beta and normal distributions, generated by different $50$ fixed for each particular dataset seeds, and we report the mean and standard error values.

\section{Results}

\subsection{FERN's Optimization Algorithms}

In order to explore how the different initial assignments, as well as the different refinement algorithms affect the final solution, we run experiments for both the refinement algorithms, starting from all four different initial algorithms presented in Subsection~\ref{sec:initialAssignment}. We evaluate the performance of the algorithms based on the final value of the objective function. The results are presented in Fig.~\ref{fig:obj_final}. 

\begin{figure}
    \centering
    \setlength{\tabcolsep}{0.001em}
    \begin{tabular}{ccc}
     \includegraphics[width=43mm]{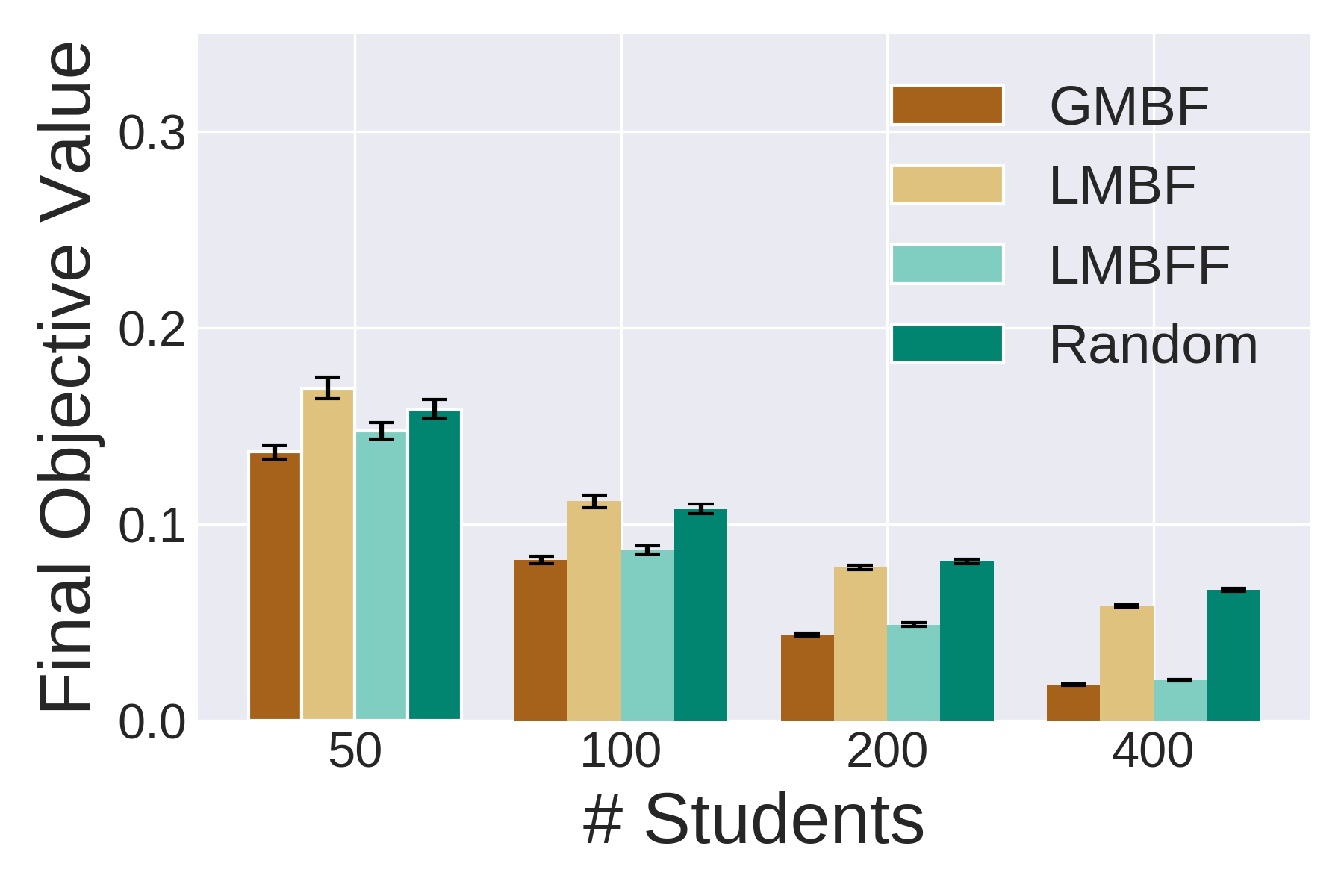} &
     \includegraphics[width=43mm]{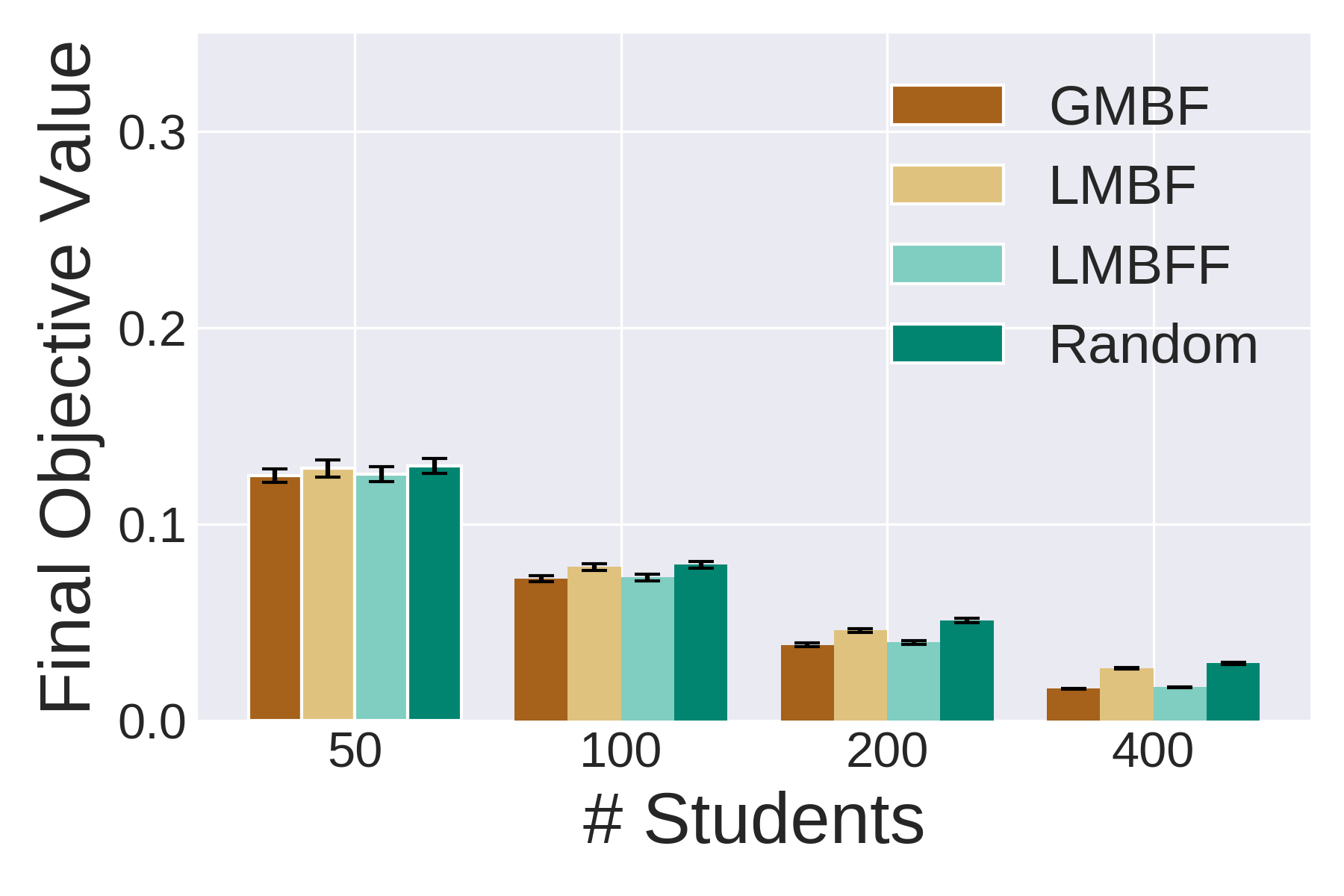} \\ 
     (a) D1 - \texttt{SAHC} & (d) D1 - \texttt{FMHC} \\
     \includegraphics[width=43mm]{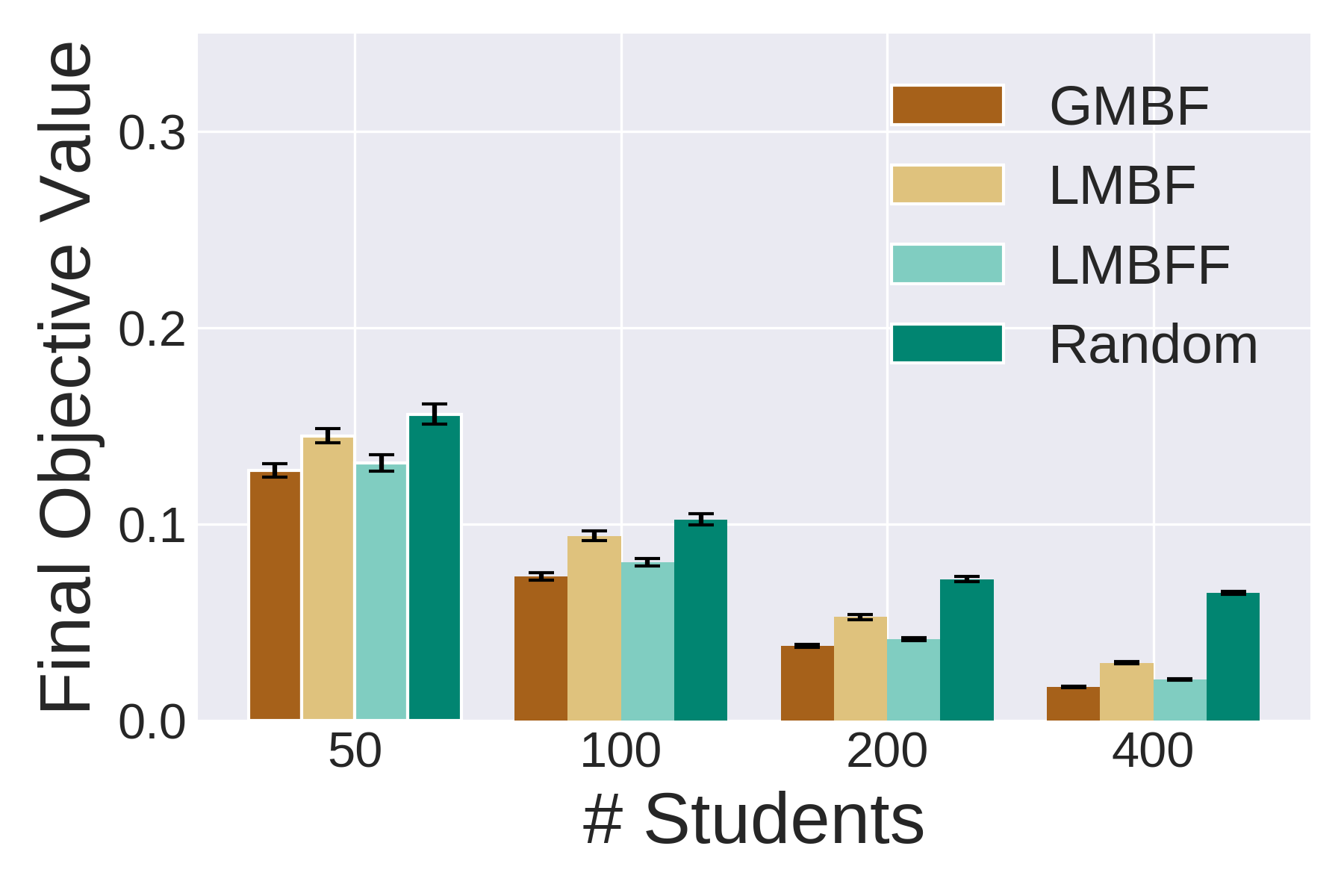} &
     \includegraphics[width=43mm]{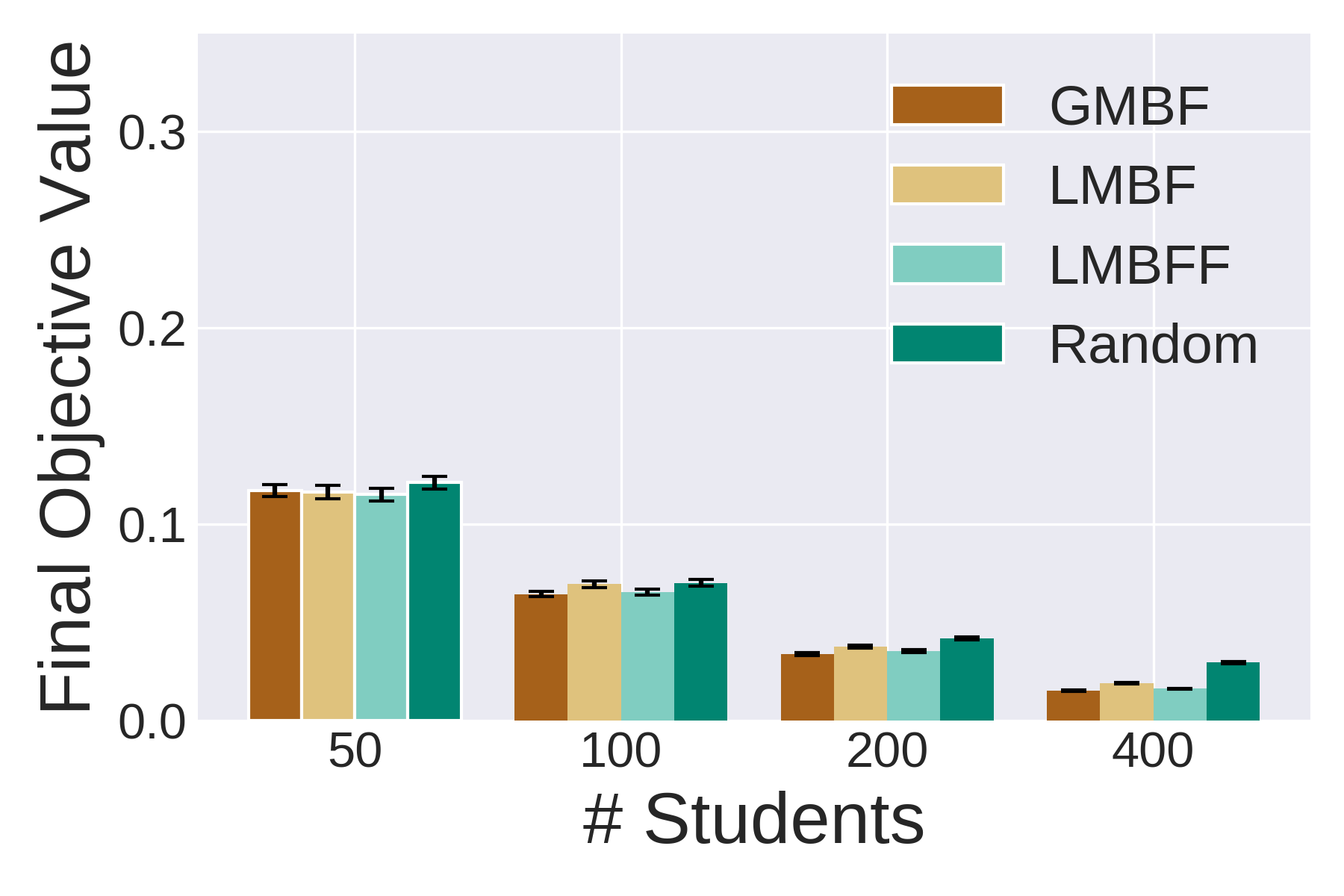} \\
     (b) D2 - \texttt{SAHC} & (e) D2 - \texttt{FMHC} \\
     \includegraphics[width=43mm]{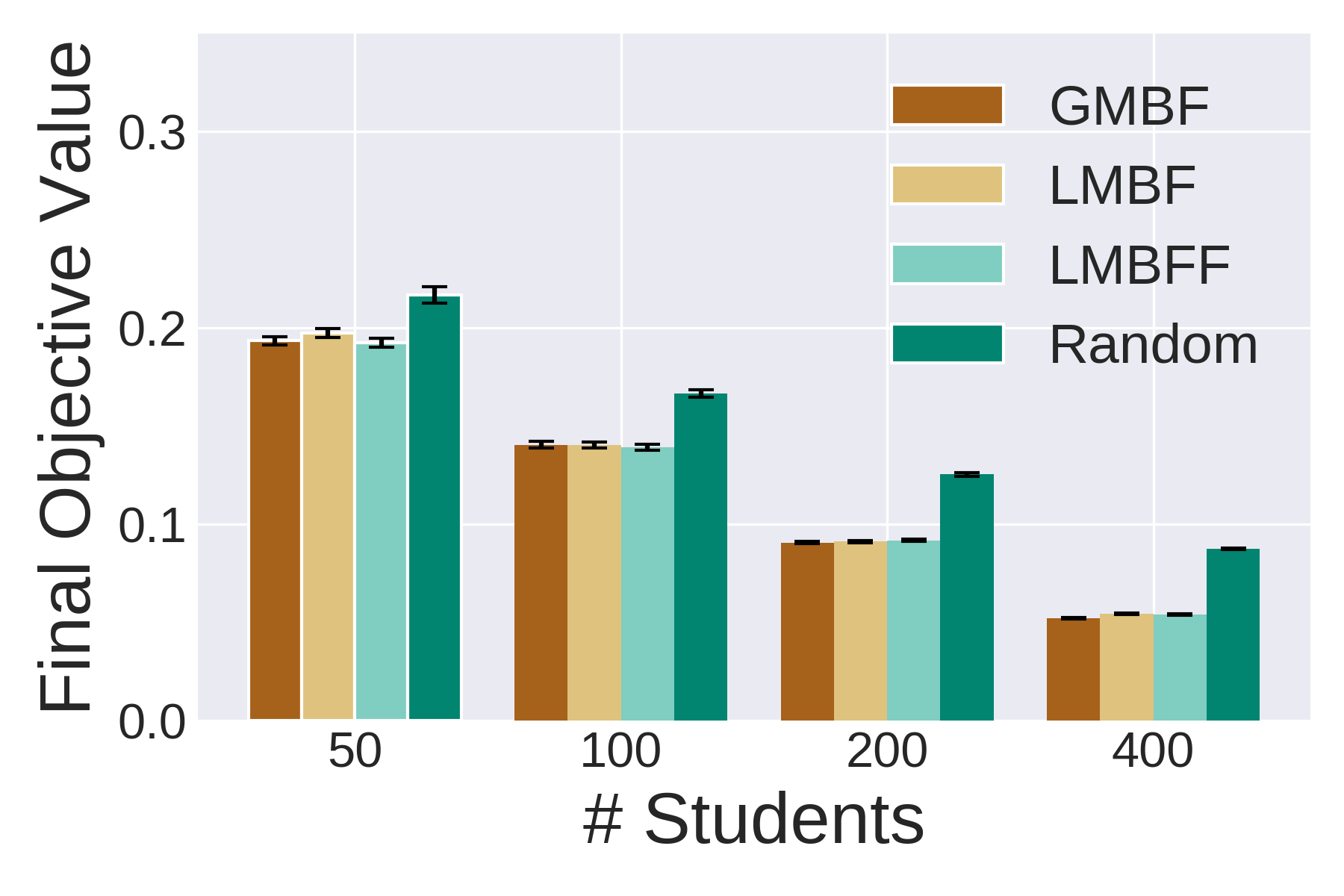} &
     \includegraphics[width=43mm]{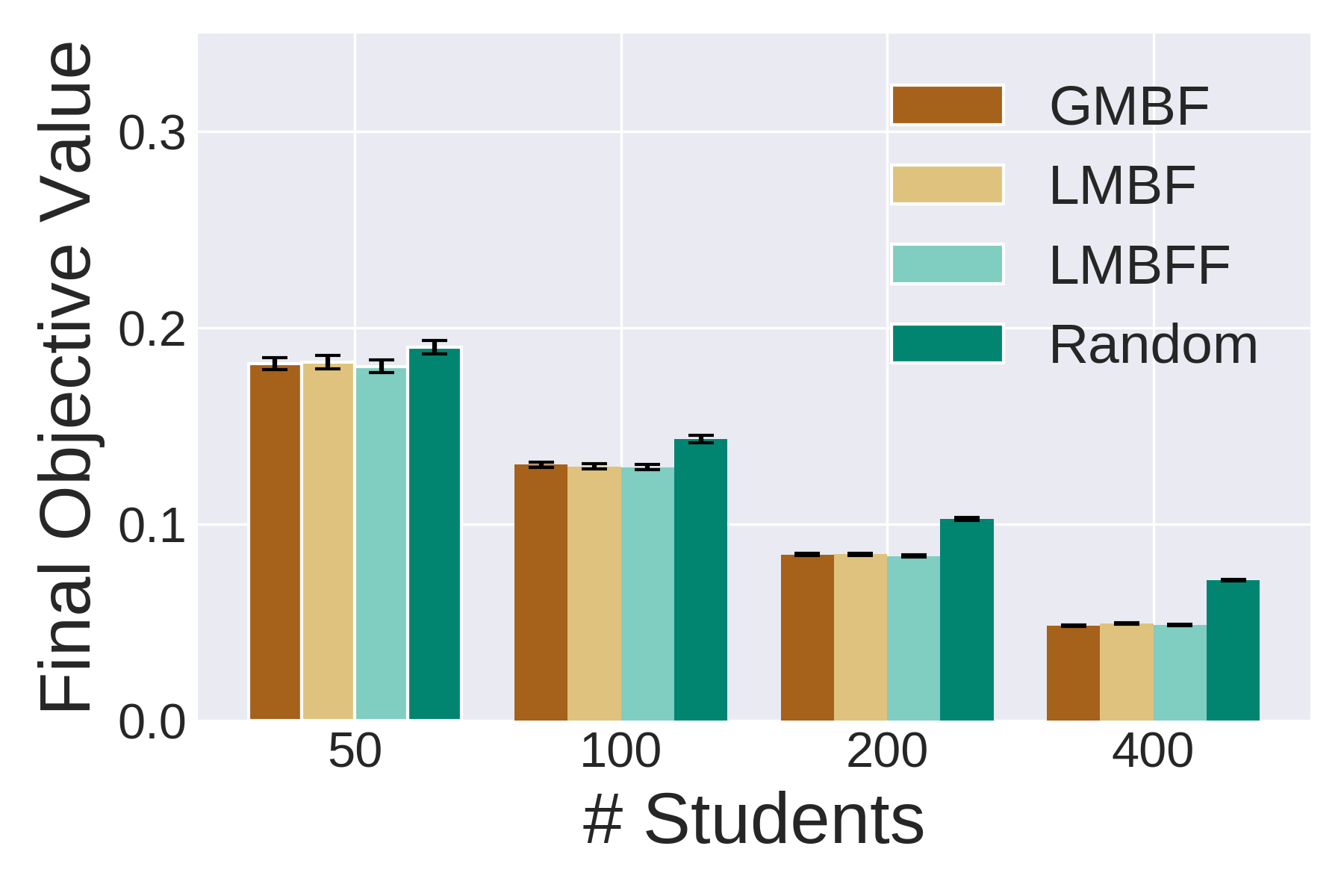} \\
     (c) D3 - \texttt{SAHC} & (f) D3 - \texttt{FMHC} \\
    \end{tabular}
    \caption{Final value of the objective function starting from all four different initial assignments for both refinement algorithms \texttt{SAHC} and \texttt{FMHC}, in all three synthetic datasets. The smallest value the better.}
    \label{fig:obj_final}
\end{figure}
Regarding the different initial assignment algorithms, \texttt{Random} performs the worst and should not be used; among the rest, there is little variation in performance. 
Regarding the different refinement algorithms, \texttt{FMHC} which allows for uphill moves, is more robust to the initial assignment and performs better than \texttt{SAHC} which gets easily trapped to local minima. 
Moreover, we see that as the number of students increases, we get better solutions. 
This is because, the more students we have, the higher flexibility we get to arrange them into teams of high benefit.
Another observation is that D3 is the most difficult one, and D1 the easiest one, verifying our assumptions on how to generate datasets of increasing difficulty.
For the subsequent experiments, we use the \texttt{GMBF} initial algorithm and \texttt{FMHC} refinement and this combination refers to \texttt{FERN}.

\subsection{Comparison Against Other Approaches}
We compare the performance of \texttt{FERN} against that of competing approaches described in Subsection~\ref{sec:baselines}, on the different datasets. The goal of this experiment is to evaluate how the different methods performed with respect to our threefold goal. In order to evaluate the performance with respect to the average skill deficiency, we report the percentage of teams that fulfilled the skill requirements. To evaluate the individual benefit of the final solutions, we report the average individual benefit (Eq.~\ref{eq:avgindben}) percentage, which measures the percentage of teammates that a student benefits from on average. Last, in order to evaluate the group fairness, we report the variance in group benefit (Eq.~\ref{eq:fairness}). 
For D1, D2 and D3, we report the results for $N=100$. Table~\ref{tab:baselines} presents the results.  

We can see that our method \texttt{FERN} is the best performing with respect to individual benefit in all datasets, as well as with respect to group fairness except for D1 in which \texttt{Umeans} and \texttt{GA} perform better. Regarding group fairness, we observe that \texttt{Umeans} is worse even than \texttt{Random}, except for D1 and MATH. This means that \texttt{Umeans} not only does not account for group fairness, but it additionally hurts it. Furthermore, D3 is the most challenging with respect to group fairness, and D2 is the second most challenging. 
Regarding the performance with respect to individual benefit, for the most difficult dataset D3, we see that on average, in \texttt{FERN}'s solution, every student benefits from $87\%$ of their teammates, compared to only $56\%$ in the case of \texttt{Umeans}. \texttt{GA} is better than \texttt{Umeans} in both individual and group benefit. However, \texttt{Umeans} is among the best performing with respect to skill deficiency, as the teams it forms are compact and concentrate adequate collective ability for the skill requirements.  
Last but not least, by comparing the performance between \texttt{FERN} and \texttt{GMBF}, we see that the refinement further improves the solution with respect to benefit.
Overall, we see that \texttt{FERN} outperforms every other method with respect to benefit and this does not come with a significant cost for the skill requirements goal.
\begin{table}[!t]
\renewcommand{\arraystretch}{1.3}
\begin{threeparttable}
	\caption{Comparison of the baselines with our method \texttt{FERN} on the different datasets.}
	\label{tab:baselines}
	\centering
	\setlength{\tabcolsep}{1.3em}
	\centering
	\begin{tabular}{llrrr}
		\hline
        Dataset & Method & Teams met skills & $\mathcal{Y}$ & $\mathcal{Z}$\\
		\hline
		\multirow{5}{*}{D1} & \texttt{Random} & 73.68 &59.23 &0.03 \\
		& \texttt{Umeans} & 95.21 &56.41 &\textbf{0.00} \\
		& \texttt{GA} & 72.05 &70.00 &\textbf{0.00}\\
		& \texttt{GMBF} & \textbf{98.21}	&86.01	&0.01 \\
		& \texttt{FERN} & 80.43	&\textbf{93.00}	&0.01\\
		\hline
		\multirow{5}{*}{D2} & \texttt{Random} & 77.98 &58.85 &164.48\\
		& \texttt{Umeans} & \textbf{99.82} &55.65 &228.22 \\
		& \texttt{GA} & 75.05 &68.79 &35.74\\
		& \texttt{GMBF} & 98.61	&86.36	&2.05 \\
		& \texttt{FERN} & 83.71	&\textbf{93.77}	&\textbf{1.47} \\
		\hline
		\multirow{5}{*}{D3} & \texttt{Random} & 75.85	&58.33	&645.47 \\
		& \texttt{Umeans} & 96.20 &56.45	&898.41 \\
		& \texttt{GA} & 78.75	&65.51	&229.13\\
		& \texttt{GMBF} & \textbf{98.74}	&80.62	&43.48 \\
		& \texttt{FERN} & 85.09	&\textbf{87.33}	&\textbf{20.61} \\
		\hline\hline
		\multirow{5}{*}{MATH} & \texttt{Random} & 75.00	&61.17	& 30.80\\
		& \texttt{Umeans} & 82.01 &83.83	&0.73\\
		& \texttt{GA} & 65.79	&69.11	& 1.67\\
		& \texttt{GMBF} & \textbf{100.00} &96.22 &2.23\\
		& \texttt{FERN} & 96.43	&\textbf{99.41} &\textbf{0.02}\\
		\hline
		\multirow{5}{*}{CSCI} & \texttt{Random} & 76.70	&59.92	&4.16\\
		& \texttt{Umeans} & 99.10 &58.43	&12.38\\
		& \texttt{GA} & 65.79	&69.11	&1.67 \\
		& \texttt{GMBF} & \textbf{100.00} &97.64	&0.01\\
		& \texttt{FERN} & 99.03	&\textbf{99.82}	&\textbf{0.00}\\
		\hline
		\multirow{5}{*}{MENG} & \texttt{Random} & 76.92	&64.13	&0.63\\
		& \texttt{Umeans} & 98.50 &64.78	&2.55\\
		& \texttt{GA} & 68.21	&75.98	&0.87\\
		& \texttt{GMBF} &\textbf{100.00}	&98.18	&0.95 \\
		& \texttt{FERN} & \textbf{100.00} &\textbf{99.75} &\textbf{0.00}\\
		\hline
	\end{tabular}
	\begin{tablenotes}
	    \item The third column (teams met skills) is the percentage of teams that fulfill the skill requirements. The fourth column ($\mathcal{Y}$) is the average percentage of teammates that a student benefits from. The last column ($\mathcal{Z}$) is the variance in group benefit (the values for group benefit correspond to percentages \%). We used $\gamma=\delta=1$ and $k=2$. The best performance is marked in bold.
	\end{tablenotes}
	\end{threeparttable}
\end{table}

\subsection{Sensitivity Analysis}
In this Section, we use \texttt{FERN} and we explore how the problem and the solution change as we scale the values of the problem's features. We mainly focus on the synthetic data, because we are interested in testing datasets of varying difficulty with respect to benefit.

\subsubsection{Distribution of the Students across the Different Protected Groups}\label{sec:distr_prot}
With this set of experiments, our goal is to determine the impact of the relative sizes of the protected groups to the difficulty of the problem.
We keep everything else fixed, except for the distribution of the students to the two protected groups. We test the following percentages: 50-50, 30-70 and 10-90. We present the results for D3 in Fig.~\ref{fig:protsize} and we report the value of the objective function.
\begin{figure}[b]
	\centering
	\includegraphics[width=0.7\linewidth]{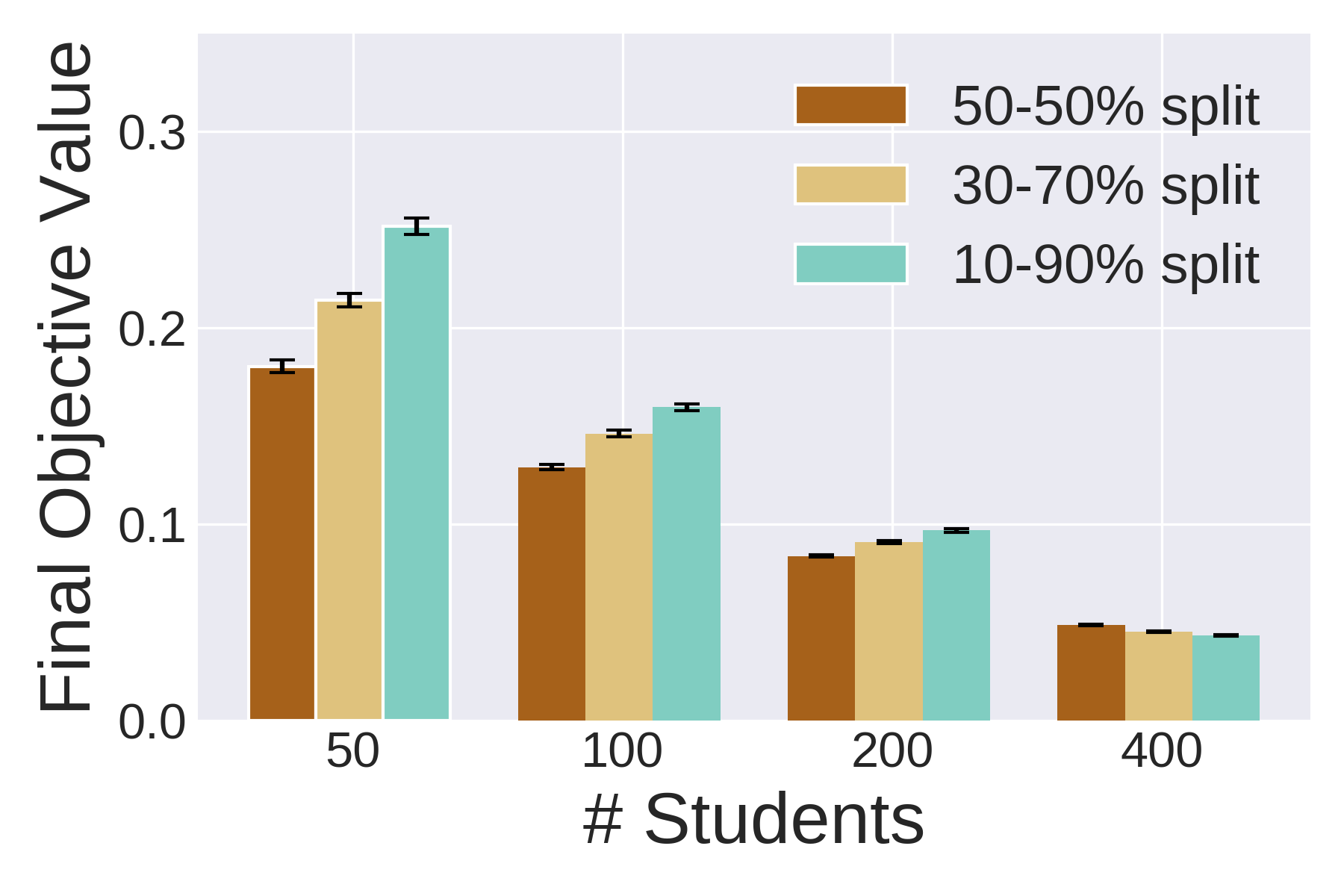}
	\label{fig:d3_protsize}
	\caption{Final value of the objective function for the different relative sizes of the protected groups, on D3. }
	\label{fig:protsize}
\end{figure}

We observe that there is not significant difference in the quality of the final solution for the different sizes of the two protected groups. 
We only show the results for D3 as the trends are similar for D1 and D2.
The only exception is for D3 and $N=50$, where there is a slight variation among the different splits. A possible explanation is that $50$ students is a relatively small input size for the algorithm to find stable solutions and is sensitive to the underlying characteristics of the data.
In conclusion, the relative sizes of the protected groups do not seem to significantly affect the difficulty of the problem. To this end, for the remaining experiments, we use equally sized protected groups.

\subsubsection{Effect of the Number of Protected Groups on the Group Benefit}
In this section, we study how the group benefit varies as we increase the number of protected groups. Fig.~\ref{fig:mul_groups} shows the results on D1, D2 and D3 for $2$, $3$, and, $4$ protected groups, when we account for the group benefit ($\delta=1$) and when we completely ignore it ($\delta=0$).
\begin{figure}
    \setlength{\tabcolsep}{0.001em}
    \begin{tabular}{cc}
     \includegraphics[width=43mm]{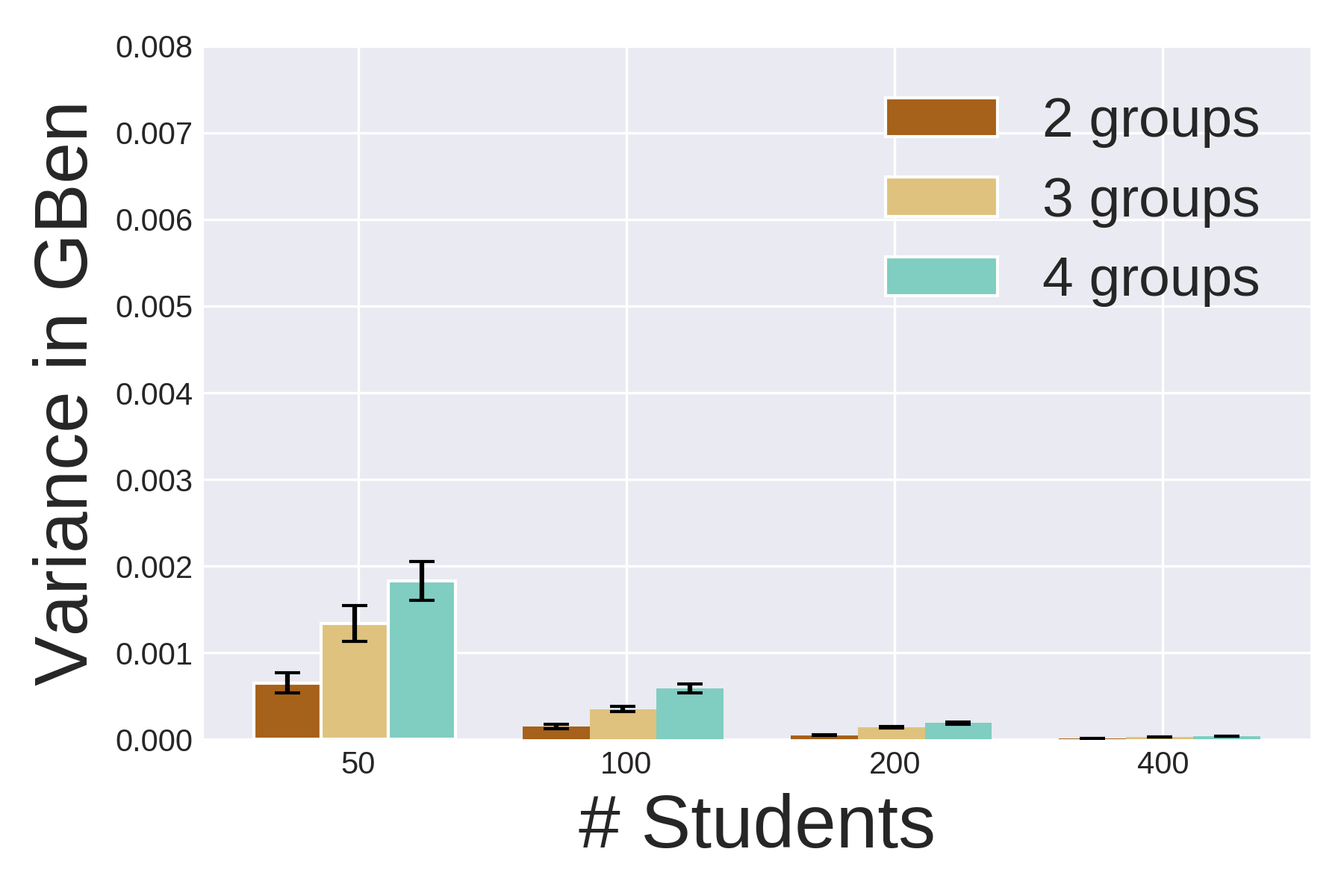} & 
     \includegraphics[width=43mm]{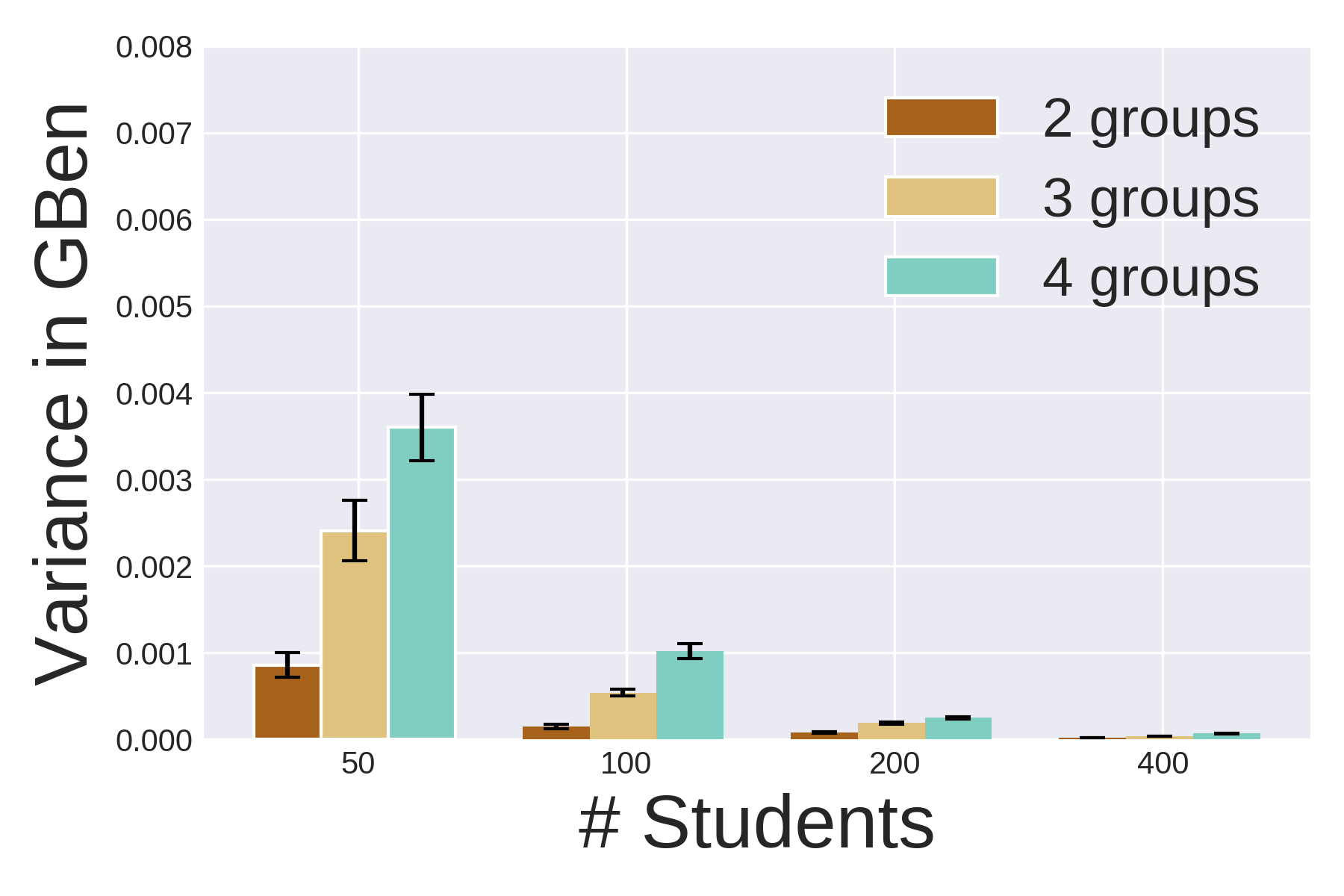} \\
     (a) D1, $\delta=1$ & (d) D1, $\delta=0$\\
     \includegraphics[width=43mm]{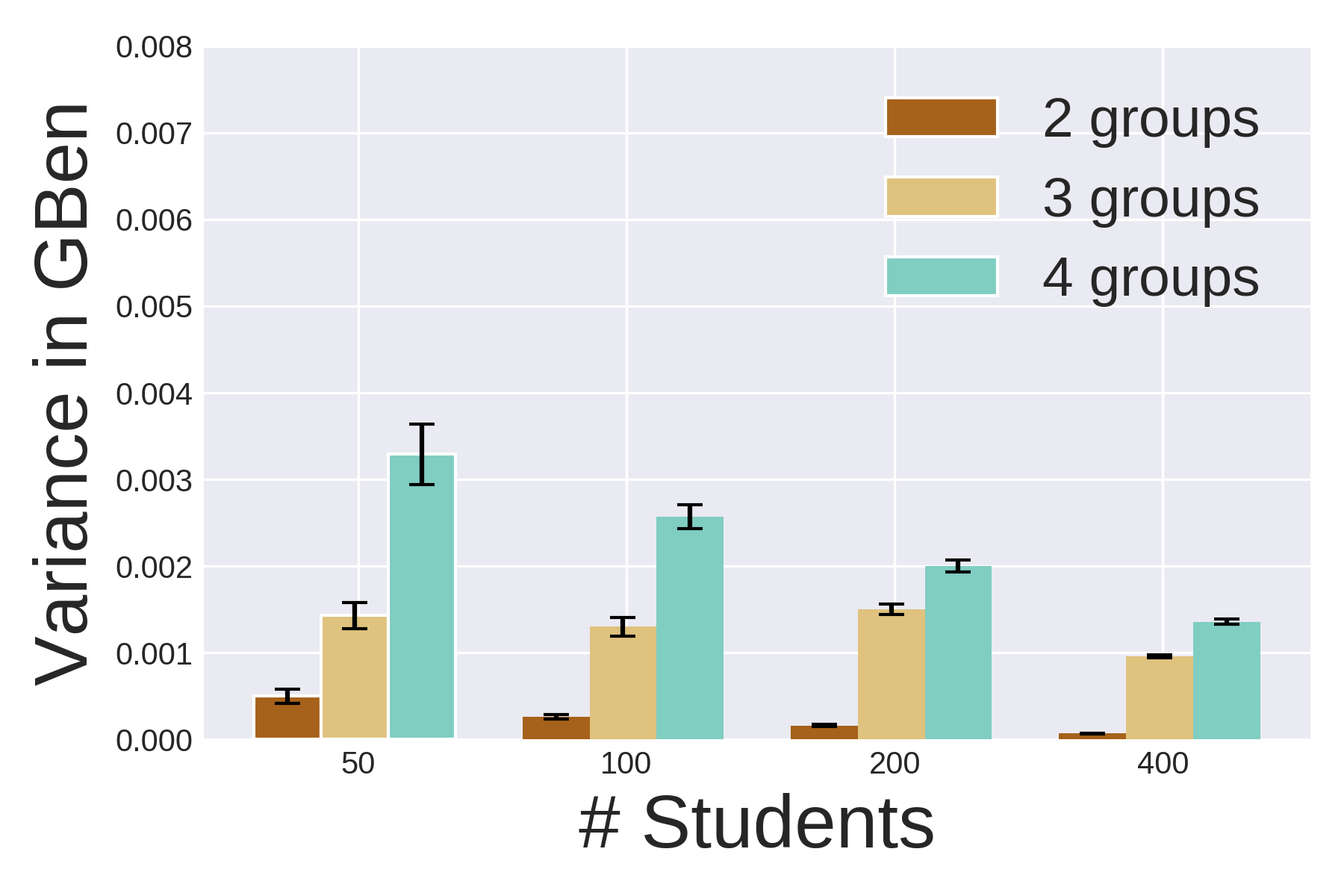} & 
     \includegraphics[width=43mm]{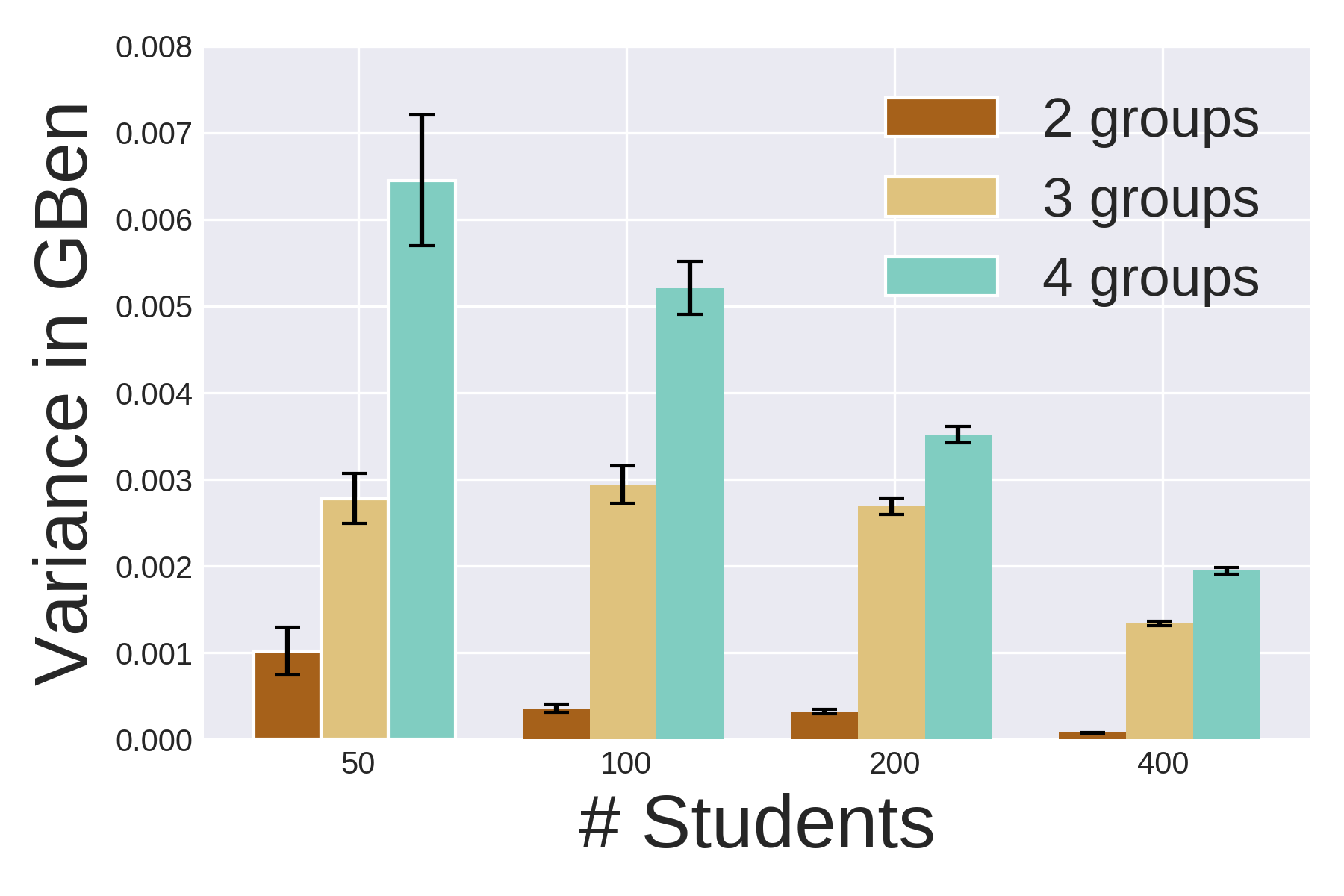} \\
     (b) D2, $\delta=1$ & (e) D2, $\delta=0$\\
     \includegraphics[width=43mm]{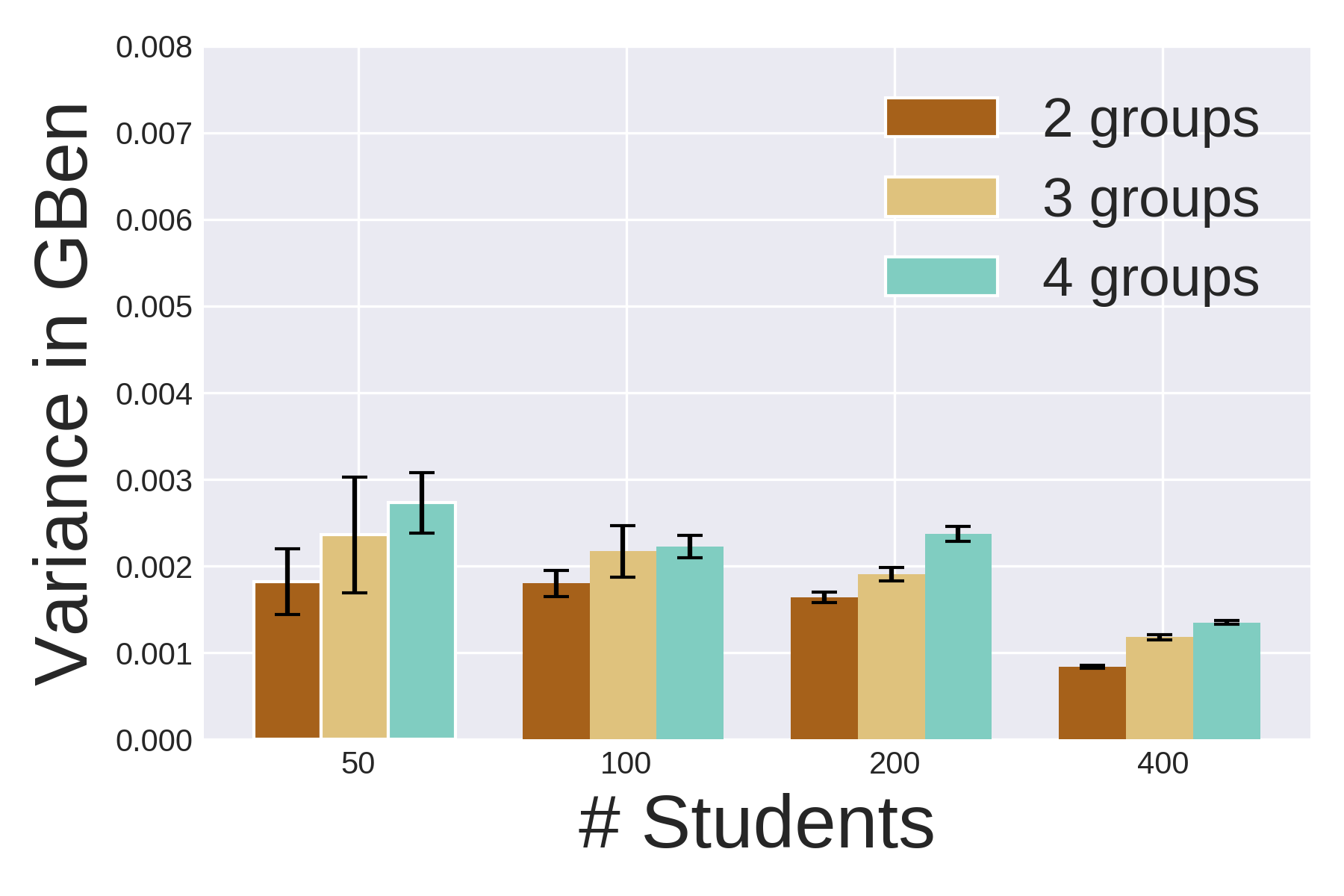} &
     \includegraphics[width=43mm]{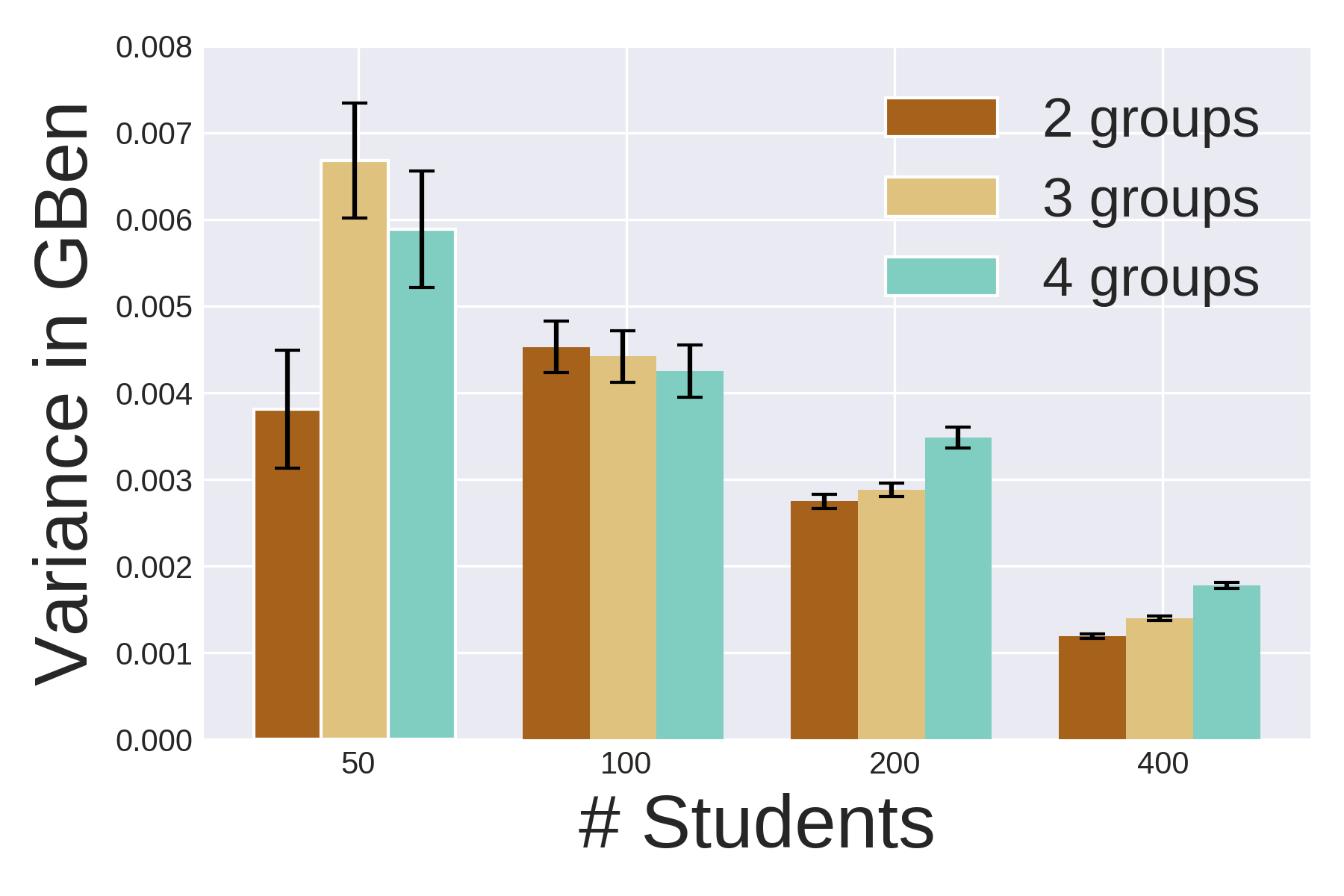} \\
    (c) D3, $\delta=1$ & (f) D3, $\delta=0$\\
    \end{tabular}
    \caption{Value of the variance in group benefit for the different numbers of protected groups, on D1, D2 and D3.}
    \label{fig:mul_groups}
\end{figure}
If we look at the results by column, we observe the following. First, as the number of groups increases, the group fairness gets more difficult to fix (with the exception of D3 for $N=50,100$ when $\delta=0$, most likely because of the extremity of the grade distribution for the 2 groups; see Fig.~\ref{fig:datasets_distr}). 
Second, as expected, the final solutions become unfair when we do not explicitly account for the group fairness (second column). This stands especially for D2 and D3, in which cases the students of each group have different grade performance and as such, it is more challenging to balance the benefit across the protected groups.    

\subsubsection{Effect of the Dimensionality of the Skill Vectors}\label{analysis-k}
We experiment with different number of dimensions for the skill vectors, $k=2,4,6$, and we evaluate the solution based on the final value of the objective function. The remaining parameters are set to the default values. We present the results for D3 in Fig.~\ref{fig:d3_k_final}. 
\begin{figure}[bt]
	\centering
	\includegraphics[width=0.7\linewidth]{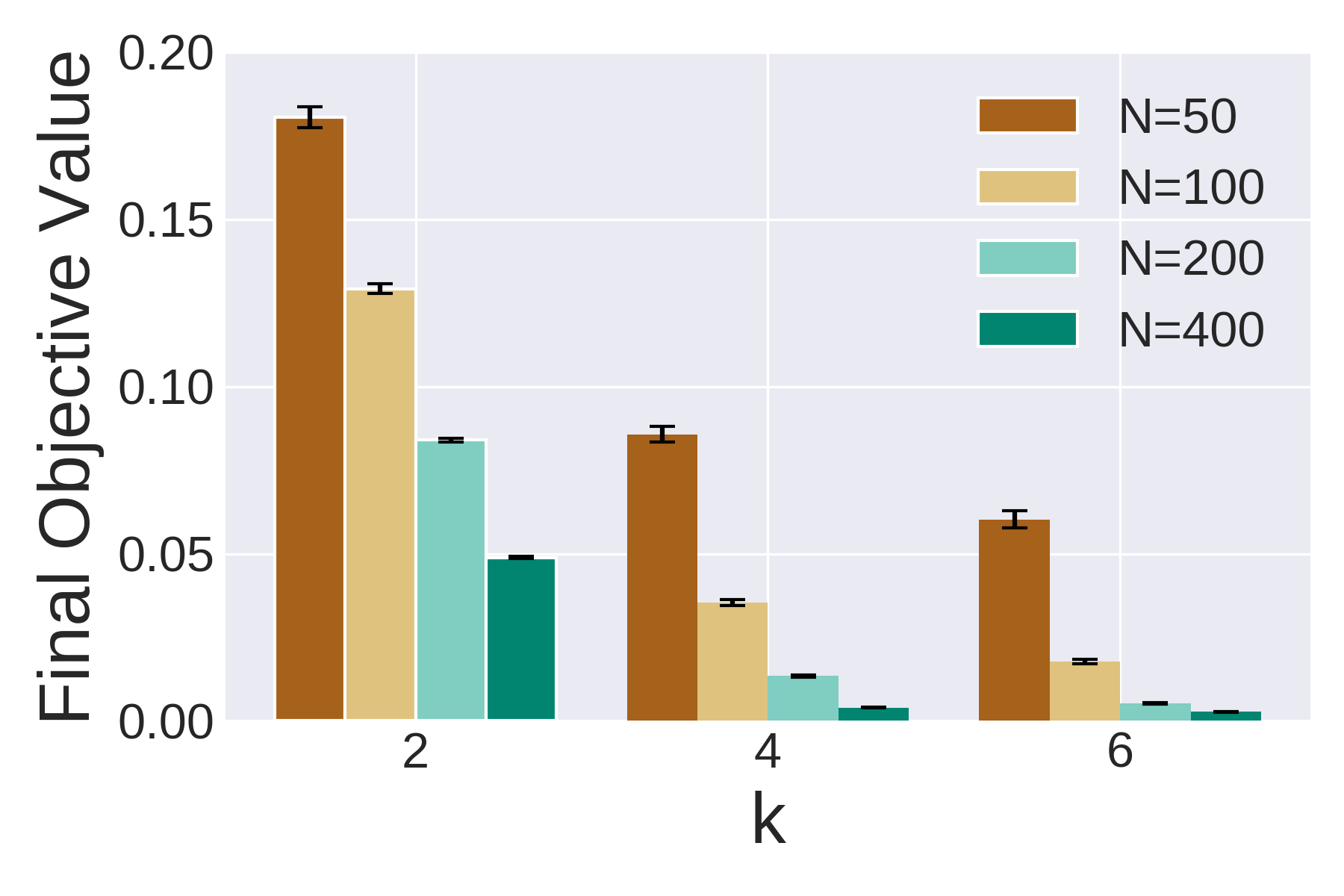}
	\caption{Final value of the objective function w.r.t. increasing number of dimensions for the skill vectors on D3.}
	\label{fig:d3_k_final}
\end{figure}
We can see that the problem gets easier as we increase the number of dimensions. This happens because, for higher values of $k$, the degrees of freedom for the skill values is higher. As a result, 
it gets easier to maximize the benefit by arranging the students properly
and the algorithm is able to find a very good and fair solution. The same results stand for the other datasets as well.

\subsubsection{Effect of the Skill Requirements Threshold}
We explore how the skill requirements vector $\mathbf{r}$ affects the difficulty of the problem and the number of teams. We run experiments for $\mathbf{r}=2,3,4$ on D3 and we set the remaining parameters to the default values.  
In Fig.~\ref{fig:d3_fm_final_r}, we plot the final value of the objective function with respect to $\mathbf{r}$. The higher the threshold, the more difficult the problem gets, as more students are needed in each team to fulfill the skill requirements and thus, it is harder to maximize the individual benefit. 
In Fig.~\ref{fig:d3_teams_r}, we plot the number of teams the algorithm generates with respect to $\mathbf{r}$. 
As we increase $\mathbf{r}$, the number of teams decreases because the average number of students per team increases. 
The trends are exactly the same for all the datasets.
\begin{figure}
	\begin{subfigure}{0.24\textwidth}
		\centering
		\includegraphics[width=\linewidth]{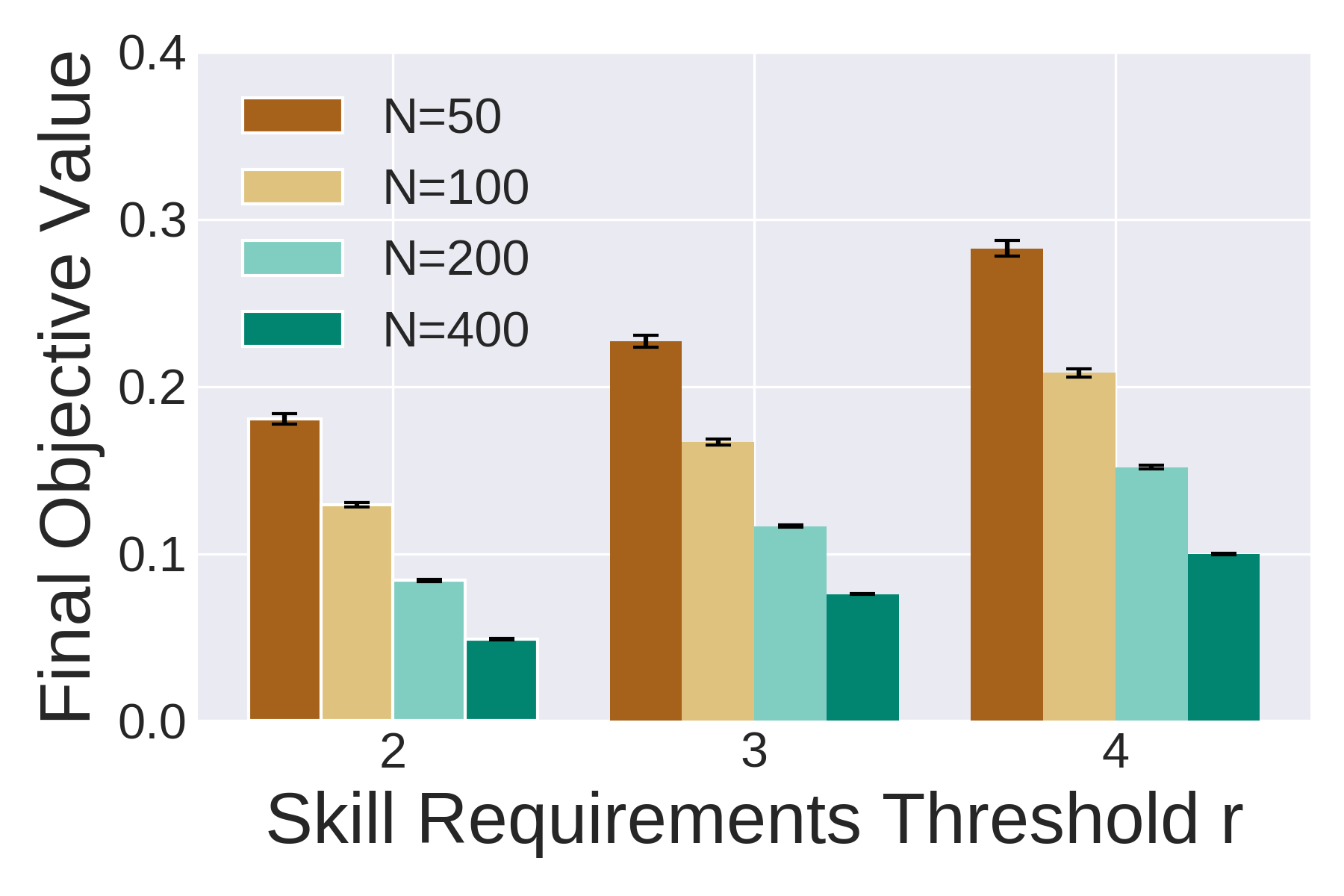}
		\caption{}
		\label{fig:d3_fm_final_r}
	\end{subfigure}%
	\begin{subfigure}{0.24\textwidth}
		\centering
		\includegraphics[width=\linewidth]{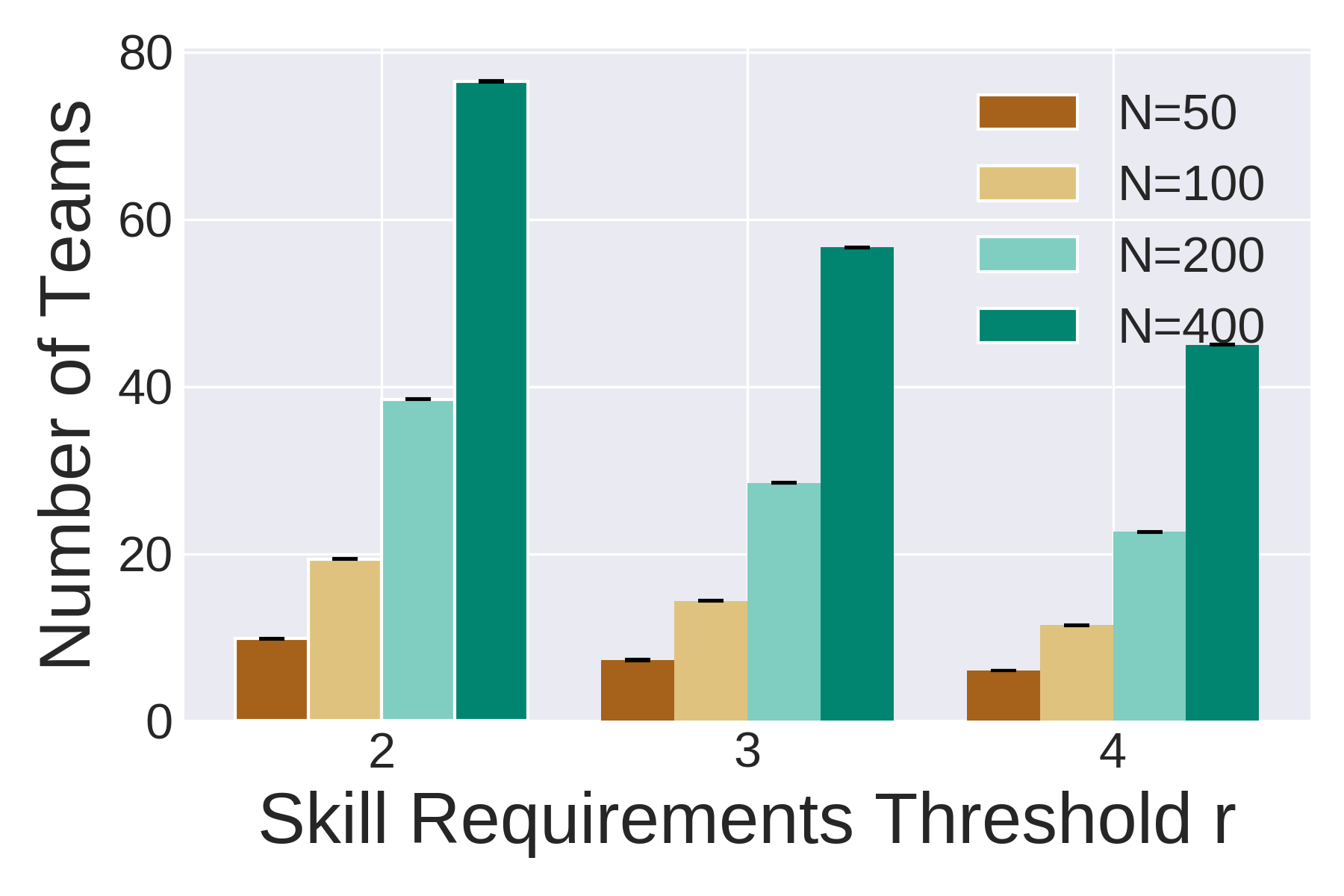}
		\caption{}
		\label{fig:d3_teams_r}
	\end{subfigure}
	\caption{Results for different values of the skill requirement threshold $\mathbf{r}$, for the different number of students in D3. (a) Final value of the objective function with respect to $\mathbf{r}$. 
    (b) Number of teams with respect to $\mathbf{r}$.}
	\label{fig:capacity}
\end{figure}

\subsubsection{Effect of each Benefit Term on the Final Solution}
How does each term of the multi-objective function affect the final solution? We answer this question by removing each term and we evaluate with the total final objective value. Fig.~\ref{fig:role_terms} presents the results.
We see that when we exclude the individual benefit ($\gamma=0, \delta=1$), the final objective value is equal to zero (optimal solution). However, this is not the case when we exclude the group benefit term ($\gamma=1, \delta=0$). 
This indicates that the individual benefit term is the hardest to optimize.

\begin{figure}[bt]
	\centering
	\includegraphics[width=0.7\linewidth]{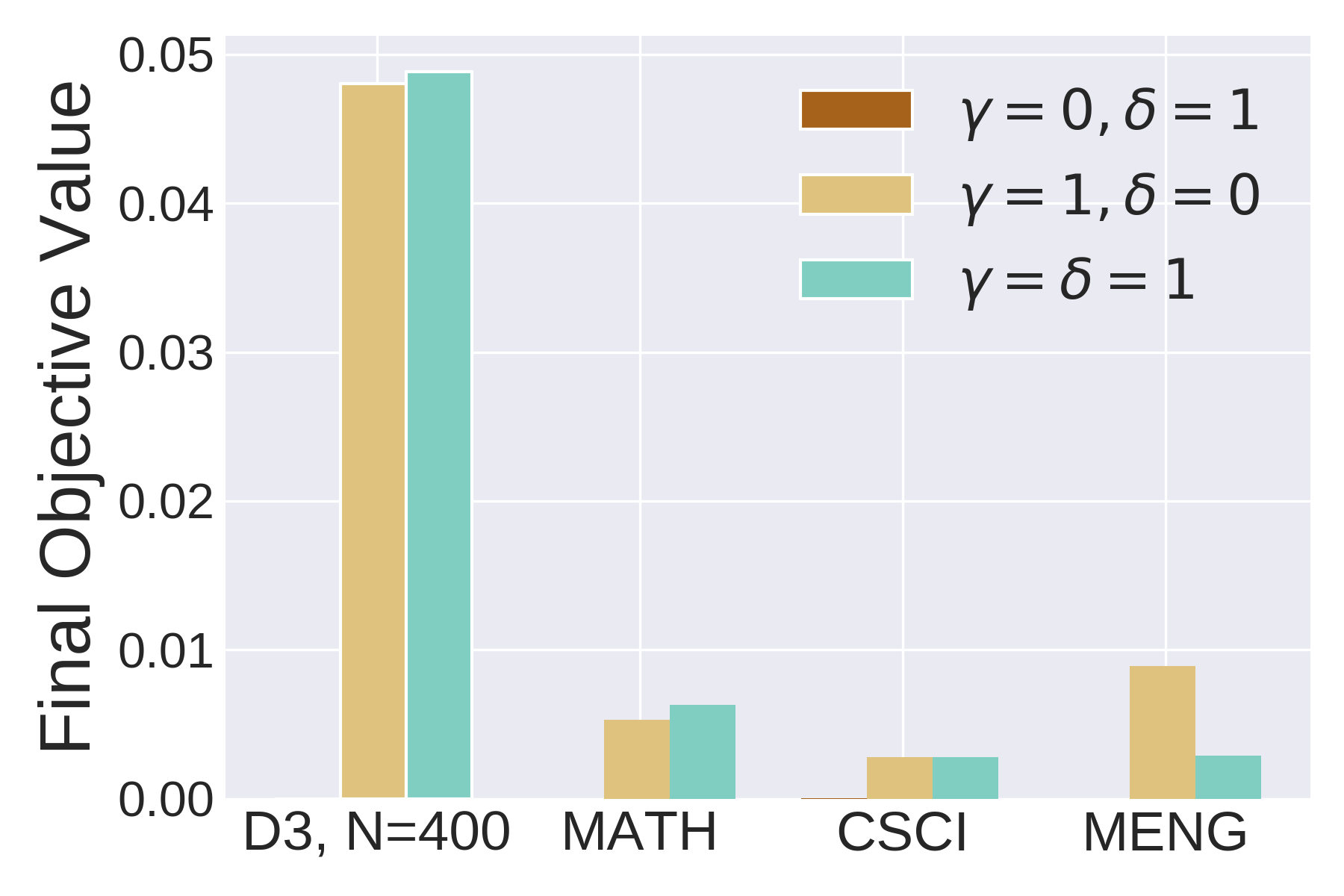}
	\caption{Final value of the objective function for the different datasets, when we exclude the individual benefit term ($\gamma=0$), when we exclude the group benefit term ($\delta=0$) and when both terms have equal contribution ($\gamma=\delta=1$).}
	\label{fig:role_terms}
\end{figure}

\subsubsection{Trade-off between Individual Benefit and Group Fairness}
Next, we study the relationship between the individual benefit and the group fairness terms, as we increase the relative importance of each one. In particular, while keeping the importance of the group term fixed ($\delta=1$), we explore how it is affected as we increase $\gamma$, i.e., as we increase the weight of the individual benefit. Looking at it the other way, this experiment also shows how the group fairness impacts the individual benefit.
Fig.~\ref{fig:tradeoff_benefits} presents the results on the most difficult dataset, D3. 

Note that we aim to maximize the individual benefit and minimize the variance in group benefit. 
We can see that when we aim to optimize the individual benefit, we end up hurting the group fairness.  
In particular, for $\gamma>1$ (Fig.~\ref{fig:tradeoff_benefits}a), the variance in group benefit deteriorates and at the same time the average individual benefit does not improve significantly.
However, this is not the case when we increase the importance of the group term (Fig.~\ref{fig:tradeoff_benefits}b). We observe that as we increase $\delta$, the variance in group benefit improves and the average individual benefit is almost constant. 
This means that we should use higher value for $\delta$ than that of $\gamma$ for the most challenging datasets D2 and D3 in order to further decrease the variance in group benefit. For the remaining datasets,
a value of $\delta>1$ does not further improve the group fairness.

\begin{figure*}
    \centering
    \begin{subfigure}{0.35\textwidth}
        \centering
        \includegraphics[width=\linewidth]{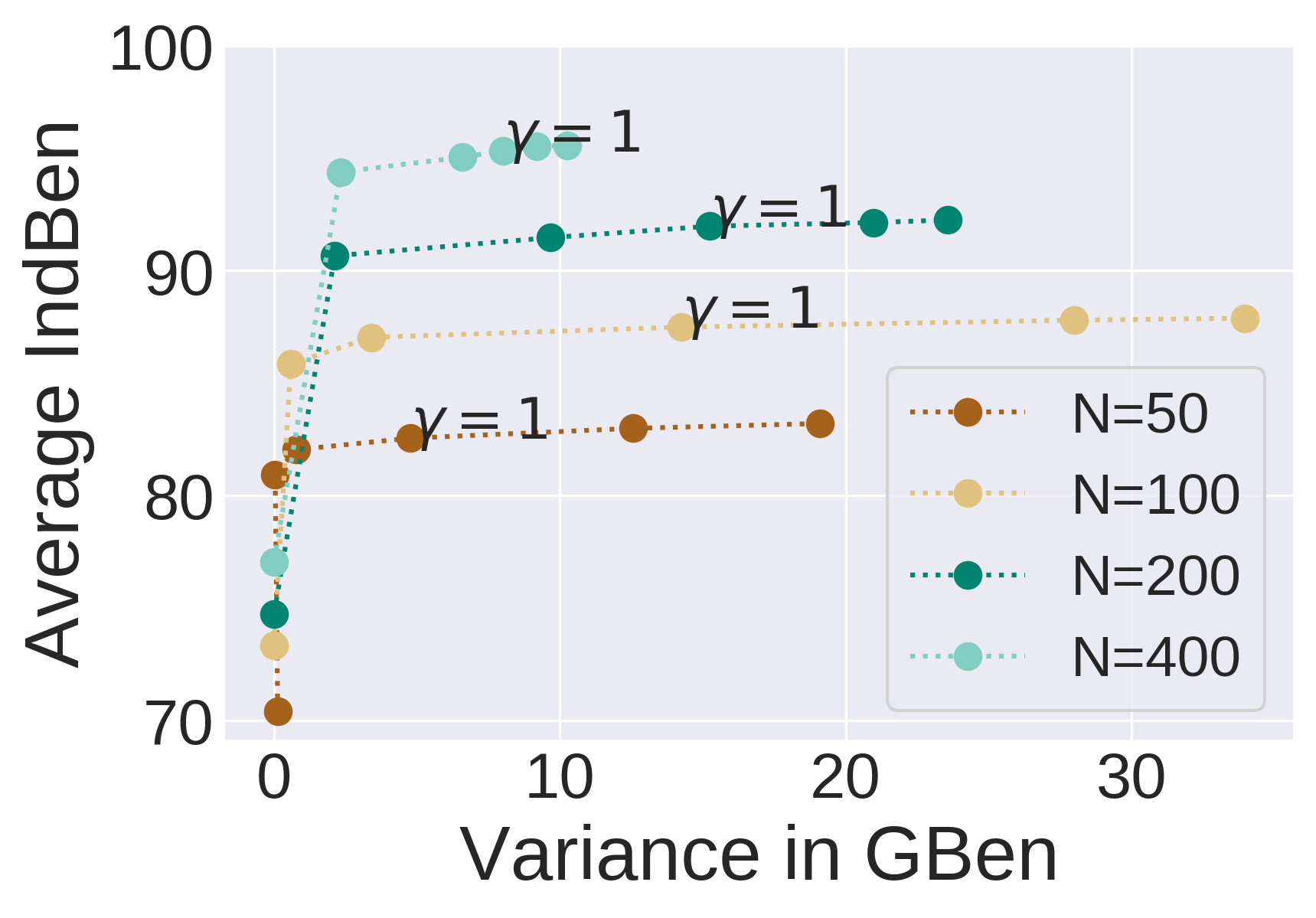} 
        \caption{$\gamma=\{0, 0.2, 0.5, 1, 2, 3\}, \delta=1$}
    \end{subfigure}\hfil
    \begin{subfigure}{0.35\textwidth}
    	\centering
        \includegraphics[width=\linewidth]{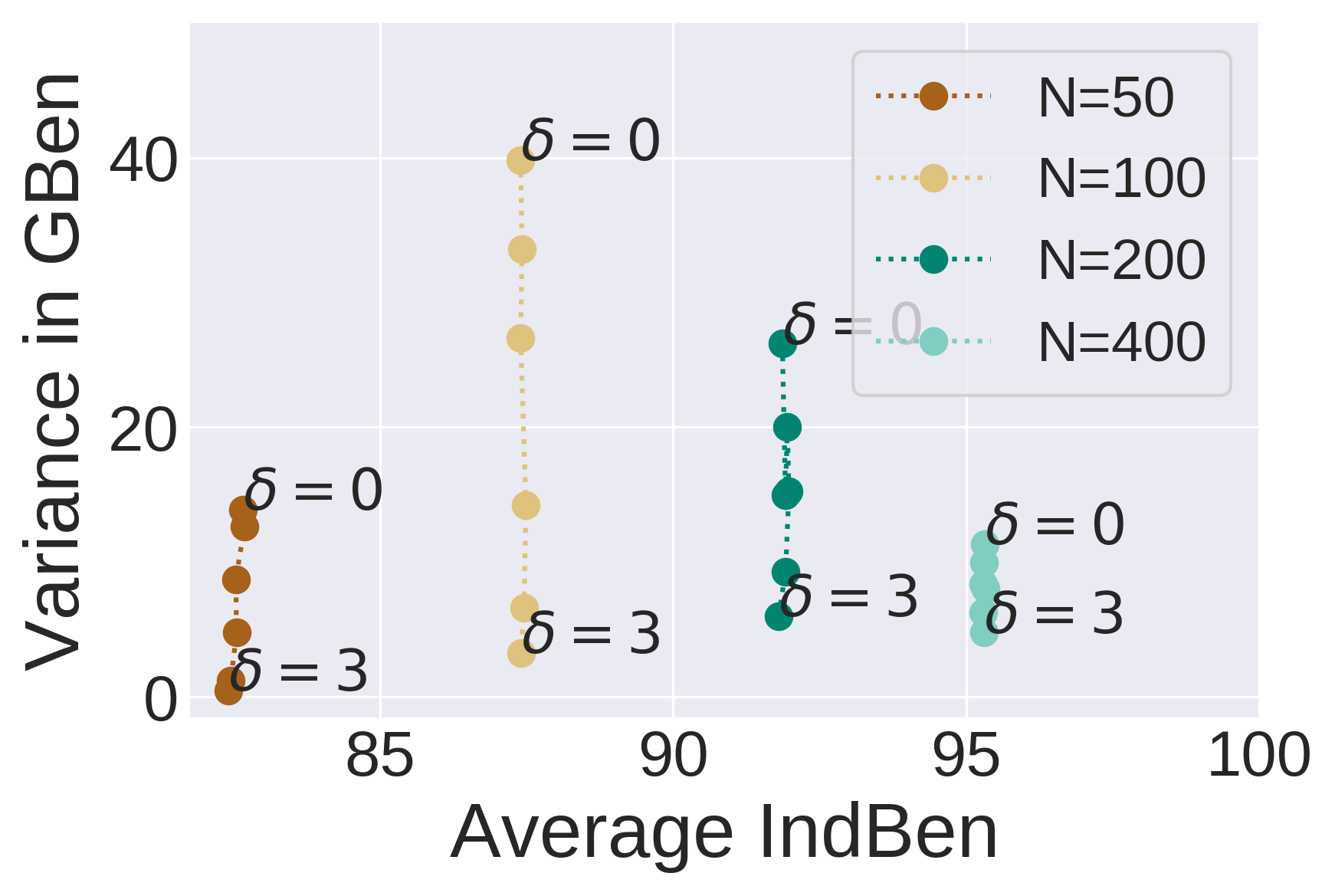} 
        \caption{$\gamma=1, \delta=\{0, 0.2, 0.5, 1, 2, 3\}$}
     \end{subfigure}
    \caption{Trade-off between the individual and the group fairness on D3, as we increase $\gamma$ (a) and as we increase $\delta$ (b). The values for \textit{Average IndBen} and \textit{GBen} correspond to percentages (\%).}
    \label{fig:tradeoff_benefits}
\end{figure*}

\subsubsection{Group Benefit of the Final Solution}
\label{sec:motivation}
 
In this subsection, we support our choice for the definition of group benefit (Eq.~\ref{eq:gben}) and the way we measure group fairness (Eq.~\ref{eq:fairness}) by showing how the group benefit changes as we change the relative importance of $\mathcal{Y}$ and $\mathcal{Z}$ in Eq.~\ref{eq:obj}. We compare the solution given by \texttt{FERN} against the one given by a group fairness unaware method, i.e., the solution by \texttt{Umeans}, on D3 and $N=400$. Table~\ref{tab:discussion} presents the results. 

According to the \texttt{Umeans} results, we can see that $q_1$ is the least benefited protected group, as its members benefit from only $26.64\%$ of their teammates, compared to $86.23\%$ of their teammates for members of $q_2$. In the first case where we ignore the individual benefit ($\gamma=0$), \texttt{FERN} increases the group benefit of the least benefited protected group $q_1$, and decreases the group benefit of $q_2$ in order to perfectly balance the group benefit and impose group fairness. In the other two cases ($\gamma=1$), the group benefit of both $q_1$ and $q_2$ increases, and $\textit{GBen}(q_1)$ in particular, increases from $26.64\%$ to about $98\%$.
Now, the difference is that when we ignore the group fairness ($\delta=0$), the gap between $\textit{GBen}(q_1)$ and $\textit{GBen}(q_2)$ is larger than the gap when $\delta=1$. The findings are similar for the other datasets as well. The second term of the objective maximizes both the individual benefit and the group benefit of each protected group. Then, the third term additionally balances the difference in group benefit.  
\begin{table}[!t]
\renewcommand{\arraystretch}{1.3}
\begin{threeparttable}
	\caption{Effect of the individual benefit and the group fairness terms to the group benefit for two protected groups on D3.}
	\label{tab:discussion}
	\setlength{\tabcolsep}{0.3em}
	\centering
	\begin{tabular}{lrrr||rrr}
	    \hline
		& \multicolumn{3}{c}{\texttt{Umeans}}&\multicolumn{3}{c}{\texttt{FERN}} \\
		\hline
        & $\textit{GBen}(q_1)$ & $\textit{GBen}(q_2)$ & $\mathcal{Z}$ & $\textit{GBen}(q_1)$ & $\textit{GBen}(q_2)$ & $\mathcal{Z}$\\
		\hline
        $\gamma=0, \delta=1$ & 26.64 & 86.23 & 887.80 & 77.01&	77.06	&0.00 \\
		$\gamma=1, \delta=0$ & 26.64 & 86.23 & 887.80 & 98.69 &91.95 &11.35\\
		$\gamma=1, \delta=1$ & 26.64 & 86.23 & 887.80 & 98.16&	92.50&	8.01\\
		\hline
	\end{tabular}
	\begin{tablenotes}
	    \item The \texttt{Umeans} results correspond to a group fairness unaware solution. $\textit{GBen}(q_1)$ and $\textit{GBen}(q_2)$ is the group benefit of protected group $q_1$ and $q_2$, respectively. $\mathcal{Z}$ is the variance in \textit{GBen}. The values for $\textit{GBen}$ represent percentages (\%). 
	\end{tablenotes}
	\end{threeparttable}
\end{table}

\section{Conclusion}
In this work we study the fair team formation problem where we maximize the benefit coming from peer learning. Given a set of students and a target task with specific skill requirements, we aim at assigning the students to teams of flexible size such that: 1) the collective ability of each team is adequate to successfully complete the target task, 2) the number of students that benefit from the team work is maximized and, 3) the different protected groups are treated fairly with respect to benefit. To the best of our knowledge, 
this is the first work that introduces fairness as \textit{equal opportunity} in collaborative learning for both the individual and protected group level. We formulate the problem as a multi-objective function and we propose \texttt{FERN}; a hill-climbing heuristic. Experimental results in both synthetic and real-world data show the effectiveness of the proposed method.

\ifCLASSOPTIONcaptionsoff
  \newpage
\fi

\bibliographystyle{IEEEtran}
\bibliography{bare_jrnl}

\end{document}